\documentclass[12pt,letterpaper]{article}

\usepackage{latexsym,epsfig,array}
\usepackage{fullpage}
\usepackage{graphicx,xcolor}
\usepackage{algorithm}
\usepackage{longtable,amsmath}

\newtheorem{theorem}{Theorem}
\newtheorem{observation}{Observation}

\newtheorem{lemma}{Lemma}
\newtheorem{corollary}{Corollary}
\newtheorem{definition}{Definition}

\newenvironment{proof}{\paragraph{Proof:}}{{\hfill $\Box$}\vspace{.5pc}}

\newenvironment{blist}
  {
\begin{list}{$\bullet$}
	{
\setlength{\partopsep}{0.00in} 
\setlength{\topsep}{0.04in} 
\setlength{\itemsep}{0.04in}
\setlength{\parsep}{0.00in}
	}
  }
{\end{list}}

\newcommand{\diam}{{\mathcal{D}}}

\begin{document}

\title{Analysis of a Memory-Efficient %Talkative
  Self-Stabilizing BFS Spanning Tree Construction\footnote{
This study was partially supported by the \textsc{anr} project \textsc{descartes : ANR-16-CE40-0023} and \textsc{anr} project \textsc{estate : ANR-16  CE25-0009-03.}
}
}

\author{
Ajoy K. Datta\thanks{Department of Computer Science,
  University of Nevada, USA . Email: Ajoy.Datta@unlv.edu.} \and
St\'ephane Devismes\thanks{Universit\'e Grenoble Alpes, 
VERIMAG, UMR 5104, 
France. Email: stephane.devismes@univ-grenoble-alpes.fr. } 
\and
Colette Johnen\thanks{Universit\'e de Bordeaux, LaBRI, 
UMR 5800, France . Email:johnen@labri.fr.} 
\and
Lawrence L. Larmore\thanks{Department of Computer Science,
  University of Nevada, USA. Email: Lawrence.Larmore@unlv.edu.}
}
\date{}
\maketitle
\thispagestyle{empty}

\bibliographystyle{alpha}

\begin{abstract}
  We present results on the last topic we collaborate with
  our late  friend, Professor Ajoy Kumar Datta (1958-2019),
  who prematurely left us few months ago.

  In this work, we shed new light on a self-stabilizing wave algorithm
  proposed by Colette Johnen in 1997~\cite{J97}.  This algorithm
  constructs a BFS spanning tree in any connected rooted
  network.  Nowadays, it is still the best existing self-stabilizing
  BFS spanning tree construction in terms of memory requirement, {\em
    i.e.}, it only requires $\Theta(1)$ bits per edge.
  However, it has been proven  assuming a
  weakly fair daemon. Moreover, its stabilization time was unknown.

  Here, we study the slightly modified version of this algorithm,
  still keeping the same memory requirement. We prove the
  self-stabilization of this variant under the distributed unfair
  daemon and show a stabilization time in $O(\diam\cdot n^2)$ rounds,
  where $\diam$ is the network diameter and $n$ the number of processes.

\end{abstract}

\section{Introduction}
We consider the problem of constructing a spanning tree in a
self-stabilizing manner.
Numerous self-stabilizing spanning tree constructions have been
studied until now, {\em e.g.}, the spanning tree may be arbitrary
(see~\cite{CYH91}), {\em depth-first} (see~\cite{CollinD94}), {\em
  breadth-first} (see~\cite{CDV09}), {\em shortest-path}~\cite{CH09},
or {\em minimum}~\cite{BPRT16j}.
Deterministic solutions to these problems have been investigated in
either fully identified networks~\cite{AKY90}, or rooted
networks~\cite{CollinD94}. 

 Here, we deal with rooted connected
networks. By ``rooted" we mean that one process, called the {\em root}
and noted $r$, is distinguished from the others. All other processes
are fully anonymous.  Such networks are something said to be {\em
  semi-anonymous}.  We focus on the construction of a Breadth-First
Search (BFS) spanning tree in such a rooted connected network, {\em
  i.e.}, a spanning tree in which the (hop-)distance from any node to
the root is minimum.

The spanning tree construction is a fundamental task in communication
networks.  Indeed, spanning trees are often involved in the design of
{\em routing}~\cite{GHIJ14} and {\em
  broadcasting}~tasks~\cite{BuiDPV99}, for
example. Moreover, improving the efficiency of the underlying spanning
tree algorithm usually also implies an improvement of the overall
solution.

We consider here the {\em atomic state model} introduced by
Dijkstra~\cite{Dij74}, also called the {\em locally shared memory
  model with composite atomicity}. In this model, the daemon
assumption accepted by the algorithm is crucial since it captures the
asynchrony of the system. More generally, self-stabilizing solutions
are also discriminated according to their stabilization time (usually
in rounds) and their memory requirement.

\paragraph{Related Work.} There are many self-stabilizing BFS spanning
tree constructions in the literature, {\em
  e.g.},~\cite{CYH91,HC92,DIM93,AB97,J97,CDV09,CRV11}.
Maybe, the first one has been proposed by Chen {\em et
  al.}~\cite{CYH91}. It is proven in the atomic state model under the
central unfair daemon and no time complexity analysis is given.  The
algorithm of Huang and Chen~\cite{HC92} is proven still in the atomic
state model, yet under a distributed unfair daemon.  In~\cite{DJ16j},
the stabilization time of this algorithm is shown to be $\Theta(n)$
rounds in the worst case, where $n$ is the number of processes.
Another algorithm, implemented in the link-register model, is
given in~\cite{DIM93}. It uses unbounded process local memories.
However, it is shown in~\cite{DJ16j} that a straightforward
bounded-memory variant of this algorithm, working in the atomic state
model, achieves an optimal stabilization time in rounds, {\em i.e.},
$O(\diam)$ rounds where $\diam$ is the network diameter.
In~\cite{AB97}, Afek and Bremler design a solution for unidirectional
networks in the message-passing model, assuming bounded capacity
links. The stabilization time of this latter algorithm is $O(n)$
rounds.  The algorithm given in \cite{CRV11} has a stabilization
time~$O(\diam^2)$ rounds, assuming the atomic state model and a
distributed unfair daemon.  All aforementioned
solutions~\cite{CYH91,HC92,DIM93,AB97,CRV11} also achieve the silent
property~\cite{DolevGS96}: a silent self-stabilizing algorithm
converges within finite time to a configuration from which the value
of all its communication variables are constant.  Two other
non-silent, {\em a.k.a. talkative}, self-stabilizing BFS spanning tree
constructions have been proposed in the atomic state model.  The
algorithm in \cite{CDV09} is proven under the distributed unfair
daemon and has a stabilization time in $O(n)$ rounds.  In~\cite{J97},
the proposed solution assumes a distributed weakly fair daemon and its
stabilization time is not investigated.

Except for~\cite{J97}, in all these aforementioned algorithms, each
process has a {\em distance} variable which keeps track of the current
level of the process in the BFS tree.  Thus, these BFS spanning tree
constructions have a space complexity in $\Omega(\log(\diam))$ bits per
process.

In contrast, the solution given in~\cite{J97} does not compute any
distance value (actually, the construction is done by synchronizing
phases).
Consequently, the obtained
memory requirement only depends on local parameters, {\em i.e.},
$\Theta(\log(\delta_p))$ bits per process $p$, where $\delta_p$ the
local degree of $p$. In other word, the space complexity of this
algorithm is intrinsically $\Theta(1)$ bits per edge. Moreover, the
algorithm does not need the {\em a priori} knowledge of any
global parameter on the network such as $\diam$ or $n$. It is worth
noticing that today it is still the best self-stabilizing
BFS spanning tree construction in terms of memory requirement.

\paragraph{Contribution.} We fill the blanks in the analysis of the
memory-efficient self-stabilizing BFS spanning tree construction given
in~\cite{J97}. Precisely, we study a slightly modified (maybe simpler)
version of the algorithm. This variant still achieves a memory
requirement in $\Theta(1)$ bits of memory per edge. We prove its
self-stabilization under the distributed unfair daemon, the weakest
scheduling assumption.  Moreover, we establish a stabilization time in
$O(\diam\cdot n^2)$ rounds, where $\diam$ is the network diameter and
$n$ the number of processes.

\paragraph{Roadmap.}
The rest of the paper is organized as follows. The computational model
is described in Section~\ref{sec:model}.  A detailed description of
the algorithm is given in Section~\ref{sec:algo}. Basic properties are
proven in Section~\ref{sec:basicProperties}. In
Section~\ref{sec:fairnessProof}, we show that every execution of our
algorithm is fair. Its stabilization time in rounds is analyzed in the
Section~\ref{sec:convergence}.

\section{Model}
\label{sec:model}
The chosen computation model is an extension of Dijkstra's original 
model for rings to arbitrary graphs \cite{Dij74}.
Consider a symmetric connected graph $G(V, E)$, 
in which $V$ is a set of processes and $E$ is a set of symmetric edges. 
We use this graph to model a distributed system with $n$ processes, 
$n = |V|$. 
In the graph, the directly connected processes are called neighbors.
Each process $v$ maintains a set of neighbors, denoted as $N(v)$. 
$N[v] = N(v) \cup \{u\}$ denotes the set of closed neighbors.
A process state is defined by its variable values.
A configuration of the system is a set of process states.

\medskip
\noindent
The proposed self-stabilizing algorithm consists of a set of 
rules. 
%Each process has several rules. 
Each rule has two parts: the guard and the action. 
The guard of a rule is a boolean function 
of the process's own state and the state 
of its neighbors. the action of a rule is a sequence of 
updating the value of the process variables.
If a rule guard on the process $v$ is verified in the 
configuration $c$, we say that $v$ is enabled at $c$.  
During a {\em computation step}, several enabled 
processes (at least one) execute a single enabled rule.
The algorithm designed ensure that at most one rule
is enabled on any process at any configuration.
An {\em execution} is a sequence of consecutive 
computations steps ($c_{1}$, $c_{2}$, ... , $c_{n}$, ... ).

\medskip
\noindent
A  set of configuration {\em A} is {\em closed} if  any 
computation step from a configuration of {\em A} reaches a 
configuration of {\em A}.
%An execution $e$ reaches  {\em A}, if $e$ has a configuration in {\em A}. 
A configuration set {\em A2} is an {\em A1-attractor},
if {\em A2} is closed and all executions starting 
from a configuration of {\em A1}, has a configuration 
of {\em A2}.

\medskip
\noindent
An algorithm {\em self-stabilizes} to {\em L} if and only if 
{\em L} is an {\em A0}-attractor  ({\em A0} being the set of  
configurations).

\medskip
\noindent
\textbf{Round complexity.}
The \textit{round} notion is used to measure the time complexity. 
The first round of a computation $e = c_1, ..., c_j, ...$ is the 
minimal prefix $e_1 = c_1, ..., c_j$, such that 
every enabled process in $c_1$ either executes a rule or 
is neutralized during a computation step of $e_1$. 
A process $u$ is \textit{neutralized} during a computation step 
if $u$ is  disabled in the reached configuration.

\noindent
Let $e'$ be the suffix of $e$ such that $e = e_1 e'$.
The second round of $e$ is the first round of $e'$, and so on.

\noindent
The stabilization time  is the maximal number of rounds needed by any
computation from any configuration to reach a legitimate configuration.

\section{Algorithm Specification}
\label{sec:algo}
We present an anonymous algorithm that builds a BFS spanning tree. 
Angluin \cite{Ang80} has shown that no deterministic algorithm 
can construct a spanning tree in an anonymous (uniform) network. 
The best that can be proposed is a semi-uniform 
deterministic algorithm, as ours, in which, all processes 
except one execute the same code. 
We call $r$, the distinguished process, the $legal~root$, 
which will eventually be the root of the BFS tree. 
\noindent
$dist(u)$ denoted the distance of process $u$ to $r$ 
in the  graph.

\medskip
\noindent
The algorithm is non-silent: 
at the end of a tree construction, the legal root initiates 
a new tree construction.
The algorithm builds $0$-colored and $1$-colored BFS spanning 
tree alternately.
The color is used to distinguish the processes of the tree 
from those that are not part of the tree: 
inly the processes in the current tree have the tree color, 
named : $r\_color$. 

\medskip
\noindent
On of the difficulty to build a BFS tree without using a 
$distance$ variable is to ensure that the path of each 
process to $r$ in the tree is minimal. 
Once the system is stabilized, 
The  trees are built in phases: during the $k$th phase, all 
processes at a distance of $k$ from $r$ join the current tree (by 
choosing a process at a distance $k-1$ from $r$ as parent)
Once $r$ has detected the end of a phase it initiates the next phase. \\
\noindent
Another difficulty to  not having a ``distance'' variable is to break cycle.
A process $u$ in a cycle that does not have the $r\_color$ detects a conflict
if one of its neighbors has the $r\_color$ but also a specific status, named $Power$.
Hence cycles are broken but not branches of the $r$-tree.

\medskip
\noindent
There are two major error handling tasks: one is to break the cycles,  
the other is to remove the illegal branches.  
The illegal roots detect their abnormal 
situation and take an $Erroneous$ status. 
The children of an $Erroneous$ process become $Erroneous$ 
illegal root. Finally, the $Erroneous$ detached processes  
recover (changing their status to $Idle$).

\medskip
\noindent
We have 2 sets of rules : the rules RC1-RC6 designed 
to ensure the illegal trees destruction, 
and to break cycles are detailed in subsection 
\ref{sub:Error-handing};
the rules R1-R7 designed to ensure the tree constructions, 
presented in subsection \ref{sub:tree-const}.
The following subsection presents the shared variables.

\subsection{Shared Variables} \label{sub:struc}

Each process $v$ maintains the following variables
($X.u$ denotes the value of $X$ in$u$ and $X.Y.u$ denotes the 
value of $X$ in process $Y.u$):

\smallskip
%\begin{blist}
%\begin{longtable} {| l | p{5.3in} |} \hline
\begin{itemize}
\item 
$TS.u$ : The parent pointer pointing to one of its neighbors or 
containing $\perp$.
$TS$ variables maintain the BFS tree structure in a distributed 
manner. 
More precisely, when the system is stabilized, 
if $u \neq r$ then  $dist(u) > dist(TS.u)$.
%\\ \hline
\item 
$P.u$ : The parent pointer pointing to one of its neighbors or
containing $\perp$.  When the system is stabilized, 
if $P.u \neq \perp$ then $P.u = TS.u$.
The variable $P$ is used to inform $P.u$ 
that the subtree construction 
rooted at $u$ is terminated
or not - if at end of a phase, $u$ has no child,
(i.e., no $u$'s neighbor has chosen $u$ as parent),
then the subtree construction rooted at $u$ is done.
%\\ \hline
\item 
$C.u$:  The color which takes value from the set $\{0, 1\}$.  
Once the system stabilizes, the processes in the current 
tree have $r\_color = C.r$ while other processes have the complement
 of $r\_color$.
%\\ \hline
\item 
$S.u$ : The status which takes value from the set 
$\{Id\-le, Working, Po\-wer, WeakE, StrongE\}$. %cj
$WeakE$ and $StrongE$ status are used during the error 
recovering process. Process $u$ has an $Erroneous$ 
status if it has the $WeakE$ or $StrongE$ status. \\
Only processes having  the $Power$ status  can have new children. 
Once the system stabilizes,
the processes at a distance of $k-1$ from $r$ only will acquire 
the Power status during the $k$th phase;
if the current phase is begun in $u$ subtree and not yet 
terminated, then $u$ has the $Working$ status.
$u$ has the $Idle$ status, if the tree construction has not reached it 
(i.e. $P.v = \perp$),
 or the current phase has not started or is finished in 
$u$' subtree.\\
$r$ can only  have the status $Power$ or $Working$.
%\\ \hline
%\item 
$ph.u$ : The phase which takes value from the set $\{a, b\}$.
The value of $S$ does not indicate if the current phase is terminated
or not.  A process in the tree is $Idle$ when it has completed
or has not started the current phase.
In order to distinguish between these two cases, we use
the phase variable. If the phase value of an $Idle$ process 
is the same as that of its parent, then the $Idle$ process has 
finished the current phase.
Otherwise, it has not initiated the current phase.
%\\ \hline
%\caption{Variables Table} 
%\label{tab:variableTable}
%\end{longtable}
\end{itemize}
%\medskip

\noindent
The root $r$ maintains the same variables, except $P$ 
and $TS$: $r$ does not
have a parent. And $S.r$ can only have the value $Power$,  
$Working$ or $StrongE$.

\medskip
\noindent
The size of $P$ and $TS$ of a process $u$ is 
$log(\delta_u)$ where $\delta_u$ is the degree of $v$.
The color, status, and phase variables have a constant 
size (total 5 bits).
Thus, the space complexity of the shared variables on $u$
is $2 \cdot \log(\delta_u)+5$  bits (i.e., $O(1)$ bits per edge).

\subsection{Recovering Rules} \label{sub:Error-handing}

A distributed system has an unpredictable initial state.
Initially, the parent pointers may point to any neighbor 
or $\perp$.
Thus, illegal trees (trees whose roots are not $r$) and
cycles (paths without a root) may exist in the initial state.
%We propose two error-handling strategies: one for eliminating 
%the illegal trees and the other for breaking the cycles. 

%\subsubsection{Cycle Path destruction process}
%\label{subsub:cyc-dest}

\begin{definition}[Cycle Path]
A series of processes $u_1$, $u_2$, ... $u_l$ is a cycle path if
$P.u_i = u_{i-1}$ for $1 < i \leq l$ and $P.u_1= u_l$. 
\end{definition}

\begin{figure}%[htb]
\begin{center}
\scalebox{0.9}{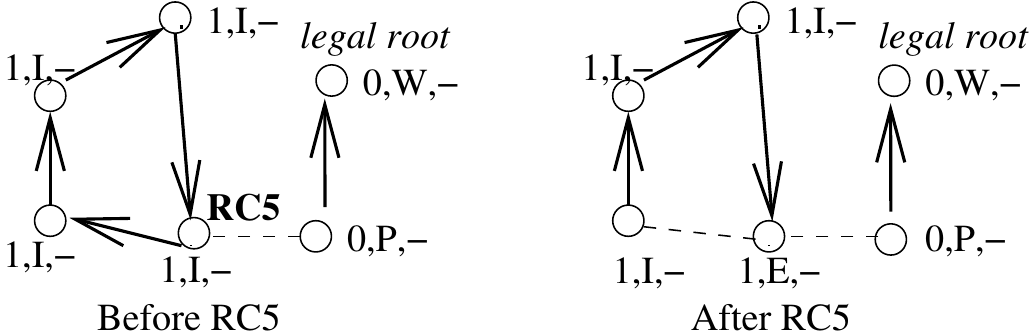}
\caption{Cycle Elimination.}
\label{fig:CycleDestruction}
\end{center}
\end{figure}

\medskip
\noindent
Processes in cycle will detect their situation with the help of their neighbors.
Once a {\em conflict} detected by a process, it  becomes
an illegal Erroneous root (rule RC4 or RC5),
 hence its cycle is transformed into an illegal tree.
A process having a parent assumes that it is in the legal 
tree and its color is equal to $r\_color$ 
(even if it is inside a cycle).
Based on this assumption, it detects a conflict when a 
$Power$ neighbor does not have its color
(both cannot be inside the legal tree).
An example of such a destruction is given in
Figure~\ref{fig:CycleDestruction}.

\begin{figure}[htb]
\begin{centering}
\scalebox{0.9}{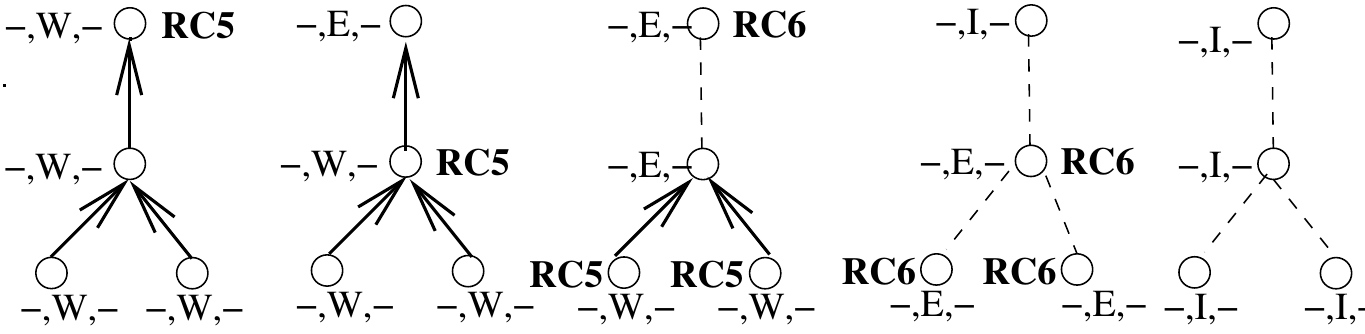}
%\makebox[\textwidth]{\hbox{\vbox{\epsfbox{treedestruction.eps}}}} % dvips
\caption{Illegal Tree Elimination.}
\label{fig:TreeDestruction}
\end{centering}
\end{figure}

\medskip
\noindent
The illegal roots detect the abnormal situation and
take $WeakE$ status by a RC5 move. 
Their children take the $WeakE$ status and quits their tree 
by a RC5 move.
Finally, the detached (i.e., without a parent and child)
$Erroneous$ processes are recovered: they change
their status by executing RC2 or RC6.
The repetition of detaching and recovering processing will correct 
all processes inside the illegal trees (see an example in
Figure~\ref{fig:TreeDestruction}).

\medskip
\noindent
If a process $u \neq r$ has in its closed neighborhood two processes 
having the Power status but not the same color 
then $u$ detects a {\em strong conflict} : the both process cannot be in 
the legal tree. Therefore, $u$ takes the $StrongE$ status (rule 
RC4).
All  $Power$ neighbors of a process having the $StrongE$ status 
verify the predicate PowerFaulty : 
they have to change their status (rule RC1, RC3, RC4, or RC5).
% (this predicate is defined in subsection \ref{subsub:misc}).
When the root detects a conflict it takes the $StrongE$ status (rule RC2).
The root detects a conflict if one of its neighbors has the $Power$ status 
but not the $r\_color$ or $r$ has not any child.

A {\em Faulty} process may become an illegal Erroneous root by execution the rule RC5.
A process is Faulty if it does not have the right color, 
the right status, or the right phase according to its parent's 
state.
More specifically, an  non-faulty process has not parent 
or it parent has  an $Erroneous$ status.
In the other cases, an non-faulty process, $u$ verifies
the $6$ following properties.
(1) $u$ has not an $Erroneous$ status.
(2) $u$ has the same color as $P.u$.
(3) $u$   status is $Idle$ if $P.u$ has not the $Working$ status.
(4) $u$ has the same phase value as $P.u$ if they have the same status.
(5) $u$ has the same phase as $P.u$ if $u$ has the $Power$ status.
(6) $u$ is childless, and $u$  has the same phase 
value as $P.u$ if $P.u$ has the $Power$ status. 
\smallskip

The following predicates are used to define the recovering guard rules:
\begin{blist}
\item 
$Child.u = \{ v \in N(u) | P.v = u\}$ 
children set contains the $u$'s neighbor direct descendent of $u$ in the tree under construction.
\item 
StrongConflict($u$) $\equiv$ $(S.u \neq StrongE)$ 
$\wedge$
\newline
\makebox[4cm]{}
$(\exists (w,v) \in (N[u],N[u])  ~|~ (C.v \neq C.u) \wedge (S.w= S.v = Power))$
\newline
The process $u$ has two closed neighbors $v$ and $w$ that 
 do not have the same color but both of them have the Power 
 status. Moreover $u$ has not the $StrongE$ status.

\item Conflict($u$) $\equiv$
  $[(u \neq r) \wedge (P.u \neq \perp) \wedge
(\exists v \in N(u) ~|~ (S.v = Power) \wedge
(C.v \neq C.u))]$ 
\newline
\makebox[1.5cm]{} \makebox[0.9cm]{}
 $\vee$ $~~~$ $[(u = r) \wedge (S.u \neq StrongE) \wedge
\newline
\makebox[4.5cm]{}
(\exists v \in N(u) ~|~ (S.v = Power) \wedge
((C.v \neq C.u) \vee (Child.u = \emptyset)))]$
\newline
The process $u \neq r$ has a neighbor $v$ which has not the $u$ color, and 
$v$ has the $Power$ status; moreover $u$ has a parent .
The root  has a neighbor $v$ which does not have $r\_color$, and 
$v$ has the $Power$ status or the root has not any child.

% so both processes should be in the legal tree and 
%should have the same color: $r\_color$.
%But, it is not the case; thus, $u$ may be inside a cycle.

%\item NotConflict($u$) $\equiv$
%$\neg$StrongConflict($u$) $\wedge$
%{\large (} $\forall v \in N(u), \neg Conflit(u,v)$ {\large )}
%\newline
%The process $u$ has not a conflict with any of its neighbor.
\item Detached($u$) $\equiv$
$(Child.u = \emptyset)$ $\wedge$ 
$((P.u = \perp) \vee (u = r)) \wedge (S.u \neq Power)$
%\makebox[0.5in]{} 
$u$ has no parent or it is the root, $u$ has no child 
and it cannot gain a child.

\item StrongEReady($u$) $\equiv$ 
 $((S.u = StrongE)\, \wedge\,
(\forall v \in N(u)  ~|~ S.v \neq Power))$ 

\item PowerFaulty($u$) $\equiv$  
$(S.u = Power)\, \wedge\, (\exists v \in N(u)  ~|~ S.v = StrongE)$

\item Faulty($u$) $\equiv$ $(u \neq r) \, \wedge\,  (P.u \neq \perp)\, \wedge$
$(S.P.u \notin Erroneous)\, \wedge\,$
\newline
\makebox[1.1in]{} {\LARGE(} 
{\large [} 
$S.u \in Erroneous$
{\large ]} $\vee$
\newline
\makebox[1.2in]{}
{\large [} 
$(C.u \neq C.P.u)$ {\large ]} $\vee$
\newline
\makebox[1.2in]{} 
{\large [} $(S.P.u \neq Working)\, \wedge\, (S.u \neq Idle)$
{\large ]} $\vee$
\newline
\makebox[1.2in]{} 
{\large [}
$(S.P.u = S.u)\, \wedge\,  (ph.u \neq ph.P.u)$ 
{\large ]}  $\vee$
\newline
\makebox[1.2in]{} 
{\large [}
$(S.u = Power) \wedge\, (ph.u \neq ph.P.u)$ 
{\large ]}
$\vee$
\newline
\makebox[1.2in]{} 
{\large [} $(S.P.u = Power)\, \wedge$ 
{\large (} $(Child.u \neq \perp) \vee
%\newline \makebox[2in]{}
(ph.u \neq ph.P.u)$ {\large )} {\large ]}  
% \newline  \makebox[1.16in]{}
 {\LARGE )}
 \item IllegalRoot($u$) $\equiv$ $(u \neq r)$ $\wedge$
$(P.u = \perp)\, \wedge$ 
$\neg$Detached($u$)
$u$ is not the legal root, it has not parent but it has a 
child or it has the $Power$ status (so it may get children).

\item IllegalLiveRoot($u$) $\equiv$ 
IllegalRoot($u$) $\wedge$ 
$(S.u \notin Erroneous)$
$u$ is an illegal root and its does not have 
an $Erroneous$ status.

\item IllegalChild($u$) $\equiv$ $(u \neq r)$ $\wedge$
$(P.u \neq \perp)\, \wedge$ 
$(S.P.u \in Erroneous)$
$u$  has an $Erroneous$ parent.

\item Isolated($u$) $\equiv$ 
 $(S.u \in \{WeakE, Working\}) \vee$ StrongEReady($u$)
\end{blist}

%\subsubsection{Recovering Miscellaneous errors} 
%\label{subsub:misc} 

%\subsubsection{Illegal Tree Elimination Process} 
%\label{subsub:ill-br}

\begin{algorithm}[thbp]
\caption{Rules for recovering on $r$.}
\begin{list}{}
{\setlength{\partopsep}{0.00in} 
\setlength{\topsep}{0.00in} 
\setlength{\itemsep}{0.04in}
\setlength{\parsep}{0.00in}} %\setlength{\leftmargin}{0.1in}}
\item{\bf RC1}~: $\neg$Conflict($r$) $\wedge$ PowerFaulty($r$) 
$\wedge$ QuietSubTree($r$)
$~\rightarrow$ $S.r := Working$;
\item{\bf RC2}~:  %$\neg$Conflict($r$) $\wedge$
Detached($r$) $\wedge$ StrongEReady($r$) 
\hspace*{2.9cm}
$~\rightarrow$ $S.r := Working$;
\item{\bf RC3}~: Conflict($r$) % $\wedge$ $(S.r \neq StrongE)$
\hspace*{7cm}
$~\rightarrow$ $S.r := StrongE$;
\end{list}
\label{fig:rRulesFaulty}
%\HRule
\end{algorithm}

\begin{algorithm}[thbp]
\caption{Rules for recovering on $u\neq r$.}
\begin{list}{}
{\setlength{\partopsep}{0.00in} 
\setlength{\topsep}{0.00in} 
\setlength{\itemsep}{0.04in}
\setlength{\parsep}{0.00in}} %\setlength{\leftmargin}{0.1in}}
\item{\bf RC4}~: StrongConflict($u$) 
\hspace*{3.5cm}
$~\rightarrow$  $S.u := StrongE$; $P.u := \perp$;
\item{\bf RC5}~: $\neg$StrongConflict($u$) $\wedge$
%$((S.u \neq WeakE)$ $\vee$ $(P.u \neq \perp))$ $\wedge$
\newline
\makebox[1.1in]{}
{\large (}Conflict($u$) $\vee$
Faulty($u$) $\vee$ PowerFaulty($u$) $\vee$
\newline
\makebox[1.2in]{}
IllegalLiveRoot($u$) $\vee$
IllegalChild($u$){\large )}
\\
\hspace*{8cm}
$~\rightarrow$ $S.u := WeakE$; $P.u := \perp$;

\item{\bf RC6}~: %$\neg$StrongConflict($u$)  $\wedge$ 
Detached($u$) 
$\wedge$  Isolated($u$) $\wedge$
\newline
\makebox[1in]{} 
( $\forall v \in N(u)$ we have
$(C.v = C.u) \vee (S.v \neq Power)$ )
$~\rightarrow$  $S.u := Idle$; 
\end{list}
\label{fig:cyc-dest-prog}
\label{fig:tree-dest-prog}
%\HRule \end{figure}
\end{algorithm}

\subsection{Tree construction rules}\label{sub:tree-const}

The rules R1 to R7 have been designed to ensure 
the tree constructions.
A R1 move initiates the tree constructions;
a R2 move initiates a phase.
R4 and R5 moves propagate the phase wave from $r$ to the 
processes in tree.
processes joint the legal tree by a R3 move.
R6 and R7 move propagates backward to $r$ the ending 
of the current phase.

\begin{figure}
\begin{centering}
\scalebox{0.9}{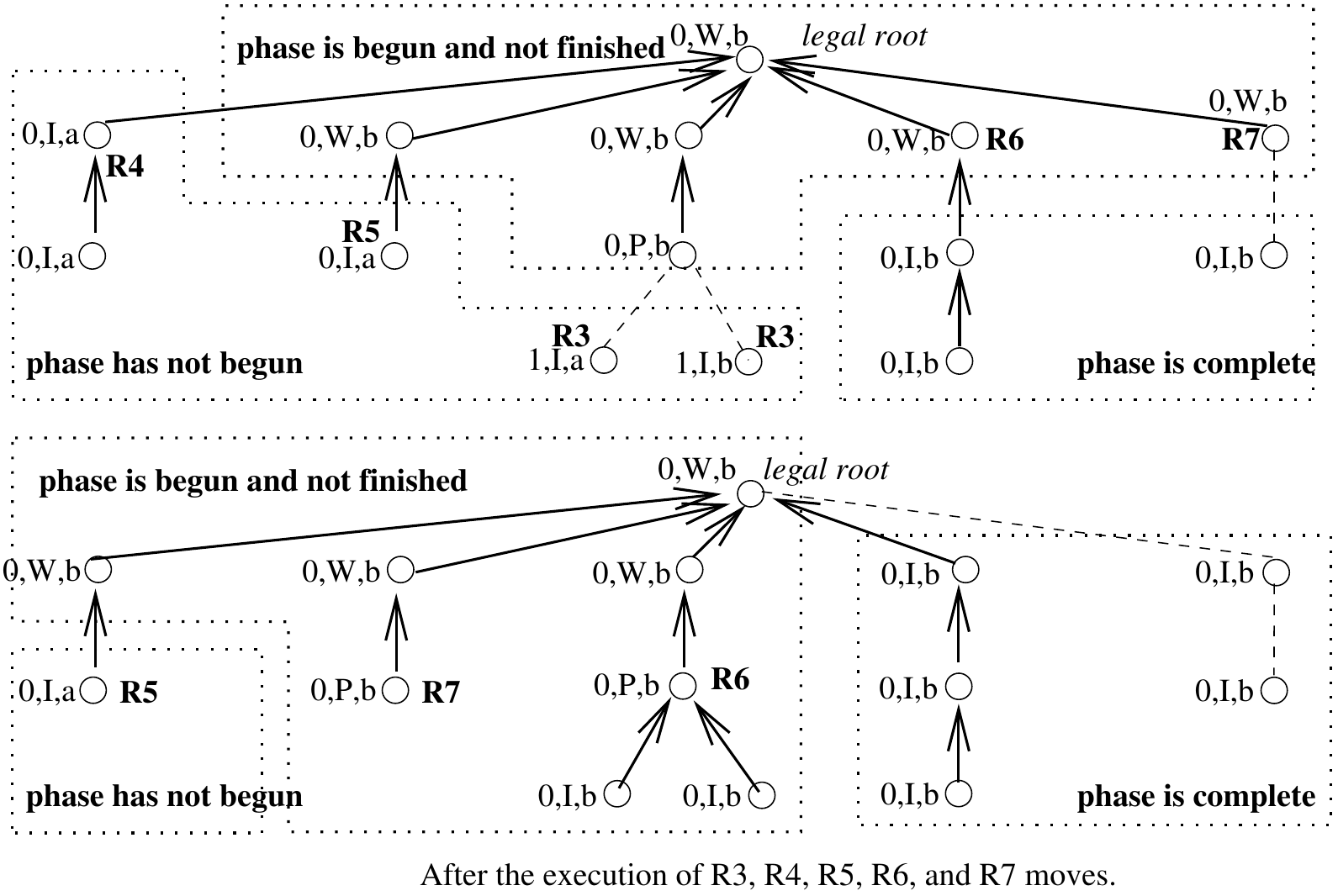}
\caption{A computation step during the 3rd phase of a
 0-colored tree construction}
\label{fig:step}
\end{centering}
\end{figure}

\medskip
\noindent
%During a kth phase, 
%All processes in the tree take the $Working$ status,
In the beginning of the $k$th phase,
all processes in the tree take
$Working$ status and $r$'s phase value (by a R4 move), 
except the leaves (processes at a distance 
of $k-1$ from $r$) which take $Power$ status (by a R5 move).
All processes at a distance of $k$ from $r$ join the tree
by choosing a $Power$ status neighbor as a parent (update their 
$P$ and $TS$ variables, but they also take the phase value and
color of the new parent) by a R3 move. 
The processes with $Power$ status will finish the $k$th phase 
(change their status to $Idle$) when the current phase is over in 
their neighborhood: all their neighbors are in the tree (they 
have $r\_color$) by a R6 or R7 move. 
The $Working$ processes will finish the phase
when their children have finished the current phase (they are 
$Idle$ and have the same phase value as them) by a R6 or R7 move.
Figure~\ref{fig:step} illustrates the computation step done during the 3rd phase of a 0-colored tree construction.
A process state is represented by a triplet: its color, its 
status, its phase, and a arrow to its parent if the process has a parent.

\begin{figure}[tbh]
	\begin{centering}
		\scalebox{0.9}{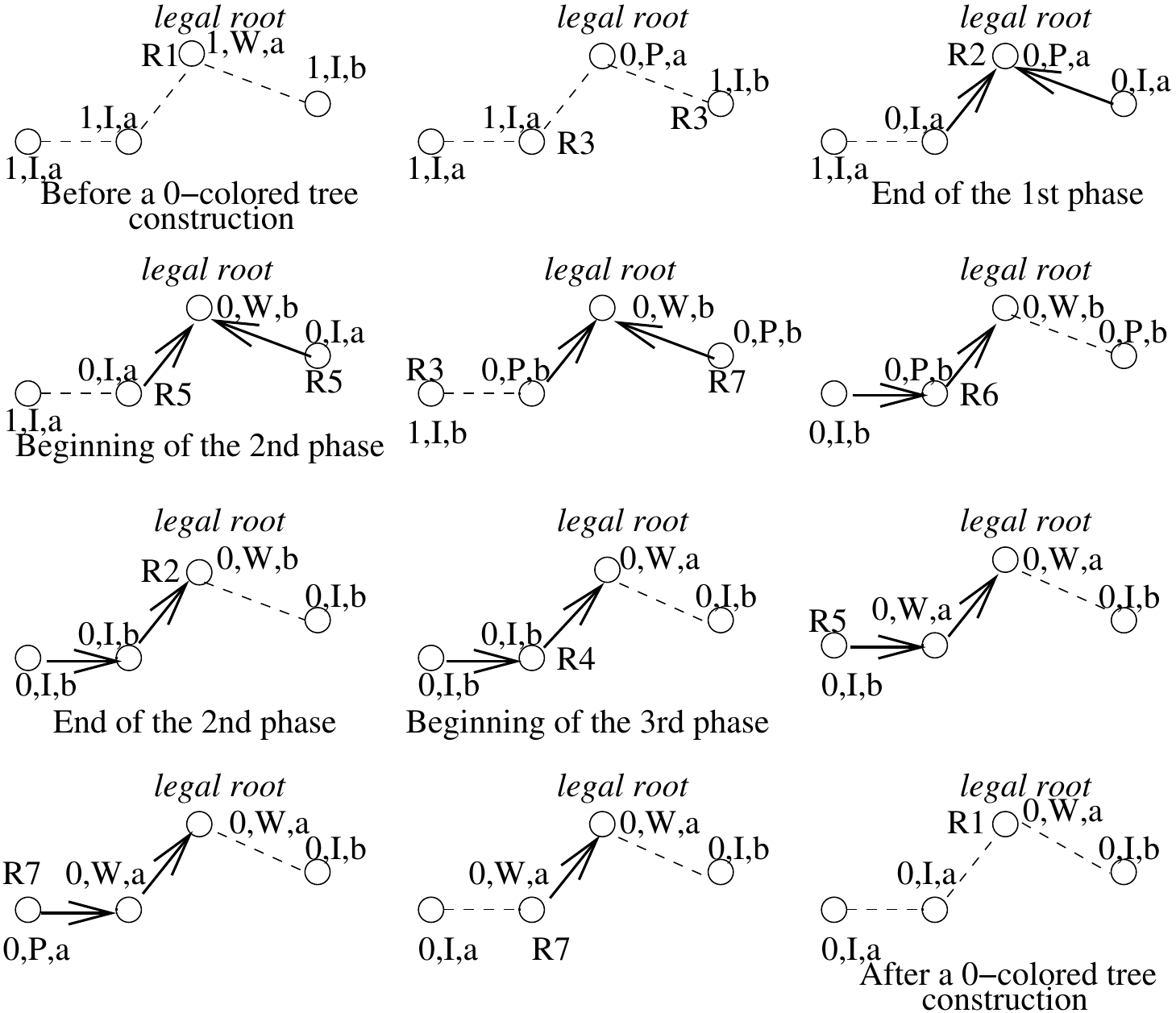}
		\caption{An 0-colored tree construction}
		\label{fig:broadcasting}
	\end{centering}
\end{figure}

\noindent
The rule R1 initiates a tree construction.
$r$ changes its color and initiates the first phase 
(by taking the $Power$ status).
When the $k$th phase is over (i.e., all processes inside 
the legal tree are $Idle$ and have the same phase as  $r$'s one),
$r$ initiates the $k+1$ phase by executing R2: 
$r$ changes its phase value.
When the subtree rooted at a process, $u$ is complete 
(i.e., its Child set is empty),
$u$ sets its $P$ variable to $\perp$ by executing R7.
Thus, when $r$ becomes childless (child.$r$ = $\emptyset$),
the tree construction is complete.
The tree is stored locally in the $TS$ variables.
%does not have anymore child 
$r$ initiates a new tree construction, by a R1 move.
A complete 0-colored tree construction is presented in 
Figure~\ref{fig:broadcasting}.

\medskip
\noindent
We define some predicates which are used in the algorithm.
\smallskip

\begin{blist}
\item Ok($u$) $\equiv$  $\neg$StrongConflict($u$) $\wedge$ 
  $\neg$Conflict($u$) $\wedge$ 
  $\neg$PowerFaulty($u$)   $\wedge$ 
  $\neg$Faulty($u$)   $\wedge$ 
\newline  \makebox [1in]{}
 $\neg$IllegalRoot($u$) $\wedge$
 $\neg$IllegalChild($u$) 
\newline
The rules RC1, RC4, RC5, and RC3 are disabled on $u$.
\item QuietSubTree($u$) $\equiv$
{\large (} $\forall v \in Child.u, 
(S.v = Idle) \wedge (ph.v = ph.u)$ {\large )}
\newline
All  children of $u$ have finished the current phase (i.e.
they are $Idle$ and have the same phase value as $u$).

\item EndFirstPhase($u$) $\equiv$
$(S.u = Power)$ $\wedge$ QuietSubTree($u$) $\wedge$
{\large (} $\forall v \in  N(u), % $|$ 
C.v = C.u$ {\large )} 
\newline
If $u$ has the $Power$ status, and all its neighbors
has its color then $u$ has terminated the current phase.

\item EndPhase($u$) $\equiv$
$(S.u = Working)$ $\wedge$ QuietSubTree($u$) 
\newline
$u$ has finished the current phase (not its first 
one so $u$ has not the $Power$ status).

\item EndLastPhase($u$) $\equiv$ 
$(Child.u = \emptyset)$ $\wedge$ {\large (}
EndFirstPhase($u$) $\vee$ EndPhase($u$) {\large )}
\newline
%\makebox [0.5in]{}
%{\em EndLastPhase(i)} holds if 
$u$ has finished a phase and $u$ is childless.  
The current tree construction is terminated in 
the subtree of $u$.
\item EndIntermediatePhase($u$) $\equiv$ 
$(Child.u \neq \emptyset)$ $\wedge$ 
{\large (} EndFirstPhase($u$) $\vee$ EndPhase($u$) {\large )}
\newline
%\makebox [0.5in]{}
%{\em EndIntermediatePhase(i)} holds if 
$u$ has finished a phase and it still has some children.  
The tree construction is not over in its subtree.
\end{blist}

\begin{blist}
\item Connection($u$, $v$) $\equiv$ (Detached($u$) $\wedge$
(Isolated($u$) $\vee$ $(S.u = Idle)$) 
$\wedge$
\newline
\makebox[1.5in]{} 
$(v \in N(u)) \wedge
(C.v \neq C.u) \wedge
(S.v = Power)$ 
\newline
%\makebox [0.5in]{}
%{\em Connection(i, k)} holds if 
$u$ is a Detached process  that has or may take the $Idle$ status.
$v$ has $Power$ status and does not have the color of $u$.
Therefore, $u$ assumes that $v$ is a leaf of the current legal tree. 
Thus, $u$ may choose $v$ as parent and take its color.
For error-recovering purpose, $u$ verifies the predicate Ok($u$).

\item NewPhase($u$) $\equiv$ $(P.u \neq \perp)$ $\wedge$
   QuietSubTree($u$) $\wedge$
   $(S.u = Idle)$ $\wedge$
$(ph.u \neq ph.P.u)$
\newline
%\makebox [0.5in]{}
%{\em NewPhase(i)} holds if 
$v$'s parent has begun a phase, but $u$ has not. More precisely,
$u$ is an $Idle$ process, 
and $u$'s phase differs from its parent phase.
Its parent has the $Working$ status (otherwise Faulty($u$) 
or IllegalChild($u$) is  verified),
%For error-recovering purpose, $u$ verifies the predicate Ok($u$).
\end{blist}
\begin{algorithm}[thbp]
\caption{Rules on $r$ for the tree construction.}
\begin{list}{}
{\setlength{\partopsep}{0.00in} 
\setlength{\topsep}{0.00in} 
\setlength{\itemsep}{0.04in}
\setlength{\parsep}{0.00in}} 
\item{\bf R1}~: $Ok(r)$ $\wedge$ EndLastPhase($r$) $\wedge$
$(\forall v \in N(u) ~|~ (S.v \neq StrongE))$ 
\newline
\makebox[1.5in]{} 
$~\rightarrow~ C.r := (C.r+1) mod2$; $S.r := Power$;
\item{\bf R2}~: $Ok(r)$ $\wedge$ EndIntermediatePhase($r$) 
\newline
\makebox[1.5in]{} 
$~\rightarrow$ $r$ changes its phase value; $S.r := Working$;
\end{list}
\label{fig:rRules}
\end{algorithm}

\begin{algorithm}[thbp]
%\begin{figure}%[htb]
%\HRule
\caption{Rules on $u \neq r$ for the tree construction.}
\begin{list}{}
{\setlength{\partopsep}{0.00in} 
\setlength{\topsep}{0.00in} 
\setlength{\itemsep}{0.04in}
\setlength{\parsep}{0.00in}} %\setlength{\leftmargin}{0.1in}}
\item{\bf R3}~: Ok($u$) $\wedge$ 
Connection($u$, $v$) 
\newline
\makebox[1.in]{}
$~\rightarrow~ C.u := C.v\,;~
ph.u := ph.v\,;~ S.u := Idle\,;~ 
P.u := v\,;~ TS.u := v$;
\item{\bf R4}~: Ok($u$) $\wedge$ 
NewPhase($u$) $\wedge~ Child.u \neq \emptyset$
$\rightarrow~ ph.u := ph.P.u\,;~ S.u := Working$;
\item{\bf R5}~: Ok($u$) $\wedge$ 
NewPhase($u$) $\wedge~ Child.u = \emptyset$ $\wedge$
$(\forall v \in N(u)$ we have $(S.v \neq StrongE))$ 
\newline
\makebox[1.in]{}
$\rightarrow~ ph.u := ph.P.u\,;~ S.u := Power$;
\item{\bf R6}~: Ok($u$)  $\wedge$ 
EndIntermediatePhase($u$) 
$\rightarrow~ S.u := Idle$;
\item{\bf R7}~: Ok($u$)  $\wedge$ $P.u \neq \perp$ $\wedge$
EndLastPhase($u$) 
$\rightarrow~ S.u := Idle\,;~ P.u := \perp$;
\end{list}
\label{fig:tree-cons-prog}
%\HRule
\end{algorithm}

\section{Basic properties}
\label{sec:basicProperties}
\subsection{Algorithm Liveness}

\begin{observation}
If StrongConflict($r$) is verified then Conflict($r$) is also verified.\\ 
$r$ cannot verified the predicates  Faulty, 
IllegalRoot, and IllegalChild.

\end{observation}
\begin{theorem}
\label{alwaysaprivilege}
In any configuration, at least one process is enabled.
\end{theorem}

\begin{proof}
%\noindent{\bf Proof} % of the theorem \ref{alwaysaprivilege}
We consider five possible global configurations~:
\begin{list}{}{
\setlength{\partopsep}{0.00in} 
\setlength{\topsep}{0.00in} 
\setlength{\itemsep}{0.04in}
\setlength{\parsep}{0.00in} \setlength{\leftmargin}{0.1in}}
\item {\bf 1} If a process $u1 \neq r$ verifying the predicate 
StrongConflict($u1$), Conflict($u1$), Faulty($u1$) or 
PowerFaulty($u1$) then $u1$ is enabled (rule RC4 or RC5). 
\item {\bf 2} All processes except $r$ verify the predicates 
$\neg$StrongConflict, $\neg$Conflict, $\neg$Faulty, 
and $\neg$PowerFaulty. $r$ verifies the predicate Conflict or PowerFaulty.
If $r$ verifies the predicate Conflict($r$) 
then the rule RC3 is enabled at $r$.
If $r$ verifies the predicates  $\neg$Conflict($r$) and
PowerFaulty($r$) then QuietSubTree($r$) is verified as
all children of $r$ verify the predicate  $\neg$Faulty. So
the rule RC1 is enabled at $r$. 
\item {\bf 3} All processes verify the predicates 
$\neg$StrongConflict, $\neg$Conflict, $\neg$Faulty, 
and $\neg$PowerFaulty. 
There is an illegal tree: there is at least a process $u \neq r$ 
verifying Illegal\-Root($u$).
\begin{list}{}{\setlength{\partopsep}{0.00in} 
\setlength{\topsep}{0.00in} 
\setlength{\itemsep}{0.04in}\setlength{\parsep}{0.00in}
\setlength{\leftmargin}{0.1in}} 
\item {\bf 3.1} If the  root  of an illegal tree does 
not have an $Erroneous$ status 
then this process can execute the rule RC5.
It verifies the predicate IllegalLiveRoot.
\item   {\bf 3.2} The roots of an illegal tree have 
an $Erroneous$ status. 
All children of an illegal root are enabled
(they verify the IllegalChild predicate).
\end{list}
\item {\bf 4} All processes verify the predicates 
$\neg$StrongConflict, $\neg$Conflict, $\neg$Faulty, 
$\neg$PowerFaulty, $\neg$IllegalLiveRoot, and 
$\neg$IllegalChild.  There is a process having an Erroneous status.
\begin{list}{}{\setlength{\partopsep}{0.00in} 
\setlength{\topsep}{0.00in} 
\setlength{\itemsep}{0.04in}\setlength{\parsep}{0.00in}
\setlength{\leftmargin}{0.1in}} 
\item {\bf 4.1}
 $S.r = StrongE$.  We have  $Child(r) = \emptyset$ 
 as no process verifies IllegalChild, and  no $r$ neighbor has the $Power$ status 
as all processes verify $\neg$PowerFaulty.
So, $r$ may execute the rule RC2.
\item   {\bf 4.2} 
Let $u4 \neq r$ be a node having an $Erroneous$ status. We have
$(P.u4 = \perp)$ as $u4$ is not faulty; $Child.u4 = \emptyset$ as 
there is not illegal tree. If $S.u4 = StrongE$ 
then no $u4$ neighbor has the $Power$ status 
(as all processes verify $\neg$PowerFaulty).
So, $u4$ may execute the rule RC6 or the rule R3.
\end{list}
\item {\bf 5} 
No process has an $Erroneous$ status. 
All processes verify the predicate Ok($u$).
\begin{list}{}{\setlength{\partopsep}{0.00in} 
\setlength{\topsep}{0.00in} 
\setlength{\itemsep}{0.04in}\setlength{\parsep}{0.00in}
\setlength{\leftmargin}{0.1in}} 
\item {\bf 5.1} The legal tree has a $Working$ leaf.
This leaf holds the R7 or R1 guard.
\item {\bf 5.2} There is an $Idle$ process $u2$
on the legal tree that does not have the same phase value 
as its parent's one.
This process verifies the R4 or the R5  guard
because the predicate QuietSubTree is verified by $u2$'s 
children as $u2$'s children are non-faulty and
no $u2$'s neighbor has an $Erronneous$ status.
\item {\bf 5.3} The legal tree has
a $Power$ process $u5$ having a neighbor $v5$ which does 
not have its color.
As, there is not conflict, $P.v5$ = $\perp$ and $S.v5 \neq Power$; 
as there is not illegal tree $Child.v5= \emptyset$.
$v5$ has the $Idle$ or the $Working$  status  as no
process has an $Erroneous$ status. $v5$ verifies the R3 guard.
\item {\bf 5.4} The legal tree has a $Power$ process $u5$ 
whose all neighbors have its color.
$u5$ verifies the R1, R2, R6, or R7 guard; 
because 
the predicate QuietSubTree is verified by $u5$'s 
children as $u5$'s children are non-faulty.
\item {\bf 5.5} The legal tree does not have
a $Power$ process, neither a $Working$ leaf. The legal tree  has 
$Working$ processes (at least the legal root).
Beside every branch of the legal tree ends by an $Idle$ leaf.
Thus, the legal tree has a $Working$ process 
whose all its children are $Idle$ (these children have its phase
otherwise see {\bf 5.2}).
This process verifies the R2 or R6 guard.
\end{list}
\end{list}
\end{proof}

\subsection{Faultless processes}

\begin{lemma}
\label{lem:FaultyIsaTrap}
Let $cs$ be a computation step from $c1$ reaching $c2$.
The predicate $\neg$Faulty is closed on process $u$.
Faulty($u$) is not verified in $c2$ if $u$ executes a rule during $cs$.
\end{lemma}
\begin{proof} 
%Let $cs$ be a computation step from $c1$ reaching $c2$.
Faulty($r$)  is never verified. In the sequel $u \neq r$

\medskip
\noindent
If $u$ executes a recovering rule (i.e. the rule  RC4, RC5, or RC6) then 
Faulty($u$) is false in $c2$ because $P.u = \perp$.

\medskip
\noindent
If $P.u = \perp$ in $c2$ then Faulty($u$) is false. 
So in the sequel, 
we name $v$ the process $P.u$ in the configuration $c2$. 
If $v$ executes the rule  RC4, RC5, or RC3  then Faulty($u$) 
is false in $c2$ because $S.P.u \in Erroneous$ in $c2$. 
As $P.u = v$ in $c2$; in $c1$ $S.v = Power$ or $u \in$ Child($v$). 
We conclude that $v$ cannot execute RC6 or RC2 during $cs$. 
Assume that $v=r$ executes the rule  RC1 in $cs$. 
 QuietSubTree($r$) is verified in $c1$, 
 the only tree construction  rule that may execute $u$ 
  during $cs$ is R3. We conclude that
$S.u = Idle$, $ph.u = ph.r$ and $ C.u = C.r$ in $c2$; so 
in $c2$  Faulty($u$) is false. 
We establish that
If $v$ executes a recovering rule (i.e. the rule  RC1, RC2, RC3, RC5, RC4, or RC6) then 
Faulty($u$) is false in $c2$ .

\medskip
\noindent
The last case to  study is the one where only tree construction rules are 
executed by $u$ and $v$ during $cs$. 
Assume that $P.u \neq v$ in $c1$. So $u$ executes R3 during $cs$, 
we conclude that in $c1$ $S.v = Power$ and $C.v \neq C.u$;
so $v$ does not execute any rule during $cs$. 
Therefore, Faulty($u$) is false in $c2$.

\medskip
\noindent
Let us study the subcase where $P.u = v$ in $c1$ and $v$ 
 executes a tree construction rule during $cs$.
$u$ cannot 
execute any tree construction rule during $cs$, because 
QuietSubTree($v$) is verified in $c1$ by hypothesis. 
If  $v$ executes a tree construction rule during $cs$
then
Faulty($u$) is false in $c2$ in regard of tree construction guards and actions.
%We conclude that Faulty($u$) is not verified in $c2$ if $v$ executes a rule during $cs$.

\medskip
\noindent
Let us study the subcase where $P.u = v$ in $c1$ and 
$u$ executes a tree construction rule during $cs$.
$v$ cannot 
execute any tree construction rule during $cs$, because 
QuietSubTree($v$) is  not verified in $c1$. 
If  $u$ executes a tree construction rule during $cs$
then Faulty($u$) is false in $c2$ in regard of tree construction guards and actions. %\\
We conclude that Faulty($u$) is not verified in $c2$ if $u$ executes a rule during $cs$.

\medskip
\noindent
Assume that $u$ and $v$ do not execute any rule during $cs$, if 
Faulty($u$) is not verified in $c1$ then Faulty($u$) is not verified in $c2$.

\medskip
\noindent
Therefore the predicate Faulty($u$) is closed.
\end{proof}

\subsection{InLegalTree, unRegular and influential processes}

A process $u$ verifies  inLegalTree if and only if $u$ is  
in a branch rooted at $r$ whose all its ascendants and $u$ are not faulty. 
Moreover $S.r$ should not be $StrongE$.\\
A process verifies either the predicate inLegalTree, Detached or unRegular.

\begin{definition}[inLegalTree and unRegular processes]
$~$
\newline
\hspace*{0.5cm}
A process $u$ verifies the predicate inLegalTree iff
\newline
\hspace*{2cm}
 $(u =r)$, and $(S.u \neq StrongE)$ or
 \newline
\hspace*{2cm}
 $(P.u \neq \perp)$,  inLegalTree is verified by $P.u$, 
  and $\neg$Faulty($u$)
\newline
\hspace*{0.5cm}
A process $u$ verifies the predicate unRegular iff
  $\neg$Detached($u$) and $\neg$inLegalTree($u$)
\end{definition}
\noindent
A PowerParent process $u$ has a descendant (maybe itself) that
may take the $Power$ status, by the execution of a series
of R4 and R5 moves. 
\begin{definition}[PowerParent processes]
A process $u$ verifies the predicate PowerParent iff
\begin {blist}
\item $(P.u \neq \perp)$, $(S.u = Idle)$, $(S.P.u = Working)$ 
and $(ph.u \neq ph.P.u)$, or
\item $(P.u \neq \perp)$, PowerParent($P.u$), 
$(S.u = Idle)$ and $(ph.u = ph.P.u)$
\end{blist}
\end{definition}

\noindent
Let $v$ be a process belonging to $Child(u)$ 
with $u$ verifying the predicate PowerParent.
The process $v$  verifies the predicate PowerParent 
 or the predicate Faulty.

\begin{definition}[influential processes]
A process $u$ verifies the predicate influential iff
\newline
\hspace*{2cm}
$(S.u = Power)$ or PowerParent($u$) 
% and $(Child.u = \emptyset)$
\end{definition}

\noindent
Any inLegalTree process is influential after the execution 
of the rule R2 by the root.
No inLegalTree process is influential 
after the execution of the rule RC1 or RC2 by the root.
Only $r$ is influential after the execution of the rule R1 
by the root.

\begin{lemma}
\label{lem:influentialMove}
Let $cs$ be a computation step from $c1$ to reach $c2$ 
where a process $u \neq r$ executes a rule.
If $u$ is influential in $c2$ then $u$ executes the rule R5.
\end{lemma}
\begin{proof} 
If $u$ executes the rule RC4, RC5, RC6, or R7, in $c2$,
we have $P.u = \perp$ 
and $S.u \neq Power$, so $u$ is not influential.

\noindent
If $u$ executes the rule R6 then we have $P.u \neq \perp$ in $c1$. 
Let us name $v$, $P.u$, in $c1$. 
We have $S.v = Working$ and $ph.v = ph.u$, in $c1$. 
On $v$, only the rule RC4, RC5, or RC3 may be enabled in $c1$, 
as Child($v$) $\neq \emptyset$ 
and the predicate QuietSubTree($v$) is not verified in $c1$. 
In $c2$, we have $S.v \neq Idle$ and $ph.u = ph.v$.
So $u$ is not influential in $c2$ as $S.u = Idle$ in $c2$.

\noindent
If $u$ executes the rule R4 then  $u$ is not influential
 in $c2$ as $S.u = Working$. 

\noindent
If $u$ executes the rule R3 then we have $P.u \neq \perp$, in $c2$. 
Let us name $v$, $P.u$ in $c2$. 
We have $S.v = Power$ in $c1$. 
On $v$, only the rule RC1, RC4, RC5, or RC3 may be enabled in $c1$, 
as the predicate EndFirstPhase($v$) is not verified in $c1$. 
In $c2$, we have $S.v \neq Idle$ and $ph.u = ph.v$. 
So $u$ is not influential in $c2$ as $S.u = Idle$ in $c2$.
\end{proof}

\begin{lemma}
\label{lem:becominginfluential1}
Let $cs$ be a computation step from $c1$ reaching $c2$. 
If during $cs$, $u \neq r$ becomes influential 
then $u$ did not execute any move during $cs$.
\end{lemma}
\begin{proof} 
Assume that $u$ executes a move during $cs$ 
and $u$ is influential in $c2$.
According to Lemma~\ref{lem:influentialMove},  $u$ executes 
the rule R5 during $cs$.
R5 is enabled only on influential process.
So $u$ does become influential during $cs$.
\end{proof}

\begin{lemma}
\label{lem:becominginfluential3}
Let $u$  be a process with $u \neq r$. 
Let $cs$ be a computation step from $c1$ reaching $c2$. 
In $c1$, $u$ is not influential but it is in $c2$.
$r$ executes R2 during $cs$, and InLegalTree($u$) is verified in $c2$. 
\end{lemma}
\begin{proof}
\noindent
Assume that in $c1$, $u_0$ is not influential, but it is in $c2$.
So $S.u_0 \neq Power$ in $c1$.

\medskip
\noindent
$u_0$  does not execute a rule during $cs$ 
(Lemma~\ref{lem:becominginfluential1}). 
We conclude that  PowerParent($u0$) is not verified in $c1$ but in $c2$.
As $u0$ does not execute a rule during $cs$, we must have $S.u_0 = Idle$ and $P.u_0 \neq \perp$ in $c1$.

\medskip
\noindent
\textbf{Base Case.}
Assume that $P.u_0 = u_1$ in $c1$ with $u_1\neq r$.

\noindent
Assume that $u_1$  executes a rule during $cs$.
To have PowerParent($u_0$), 
we must have $S.u_1 \in \{Working, Idle\}$ in $c2$.
As Child($u_1$) $\neq \emptyset$ in $c1$, 
$u_1$ would execute the rule R4 or R6.
If $u_1$ had executed R6 then in $c1$ we would have $Ok(u_1)$,
 $S.u_1 \in \{Working, Power\}$ and $ph.u_1 = ph.P.u_1$. Let us name $u_2$
process $P.u_1$ in $c1$. 
During $cs$, $u_2$ may only execute the rule RC1, RC2,  RC4,  or RC5.
In $c2$, we would have $P.u_1 = u_2$, $ph.u_1 = ph.u_2$ 
and $S.u_2 \neq Idle$; so PowerParent($u_1$) would not be verified in $c2$.
If $u_1$ had executed R4 then 
PowerParent($u_1$), and PowerParent($u_0$)  would be verified in $c1$: 
a contradiction with the hypothesis.
We conclude that $u_1$  does not execute a rule during $cs$. 

\noindent
We have $S.u_1 = Idle$ and $ph.u_1=ph.u_0$ 
in $c1$ otherwise  $u_0$ is influential in $c1$ or
$u_0$ is not influential in $c2$. 
Notice that 
PowerParent($u_1$) has to be verified in $c2$ to ensure that  
PowerParent($u_0$) is verified in $c2$; but PowerParent($u_1$) 
is not verified in $c1$ as PowerParent($u_0$) is not verified in $c1$.

\medskip
\noindent
\textbf{By induction}, we  establish that  a finite series of processes 
$u_0$, $u_1$, $\cdots$, $u_l$ verifying the following properties exists:
PowerParent($u_i$) is not verified, $P.u_i = u_{i+1}$, 
$S.u_i = Idle$, and $ph.u_i=ph.u_{i+1}$
for $0 \leq i < l$ in $c1$. 
We have either  $P.u_l = \perp$, $u_l = u_i$ with $i <l$, 
or $u_l = r$ in $c1$.

\noindent
First case, we have $P.u_l = \perp$ or $u_l = u_i$ with $i <l$.
According to the \textit{base case} paragraph, 
processes $u_i$ with $0 \leq i \leq l$ do not execute 
any action during $cs$. %and we have $S.{u_l} = Idle$ in $c1$.
We conclude that $u_0$ is not influential in $c2$. 

\noindent
Last case $u_l = r$. 
During $cs$, processes $u_i$ with $0 \leq i < l$ 
do not execute any action (see \textit{base case} paragraph). 
So inLegalTree($u_i$) is verified in $c2$.
Only if $r$ executes R2, PowerParent($u_i$) is  verified 
for $0 \leq i < l$ in $c2$. 
\end{proof}

\begin{corollary}
\label{cor:preInf}
Let $cs$ be a computation step from $c1$ reaching $c2$. 
If during $cs$, $v$ becomes influential then $C.v=C.r$ in $c2$.
\end{corollary}

\subsection{UnSafe processes}
\noindent
A unSafe process $u$ verifies the predicate inLegalTree, $u$ may 
have influential descendants, and $u$ may execute RC4 or RC5. 
After this move, its descendants have become unRegular.

\begin{definition}[insideLegalTree and unSafe processes]
$~$
\newline
\hspace*{0.5cm}
A process $u$  verified the predicate insideLegalTree($u$) iff 
\newline
\hspace*{2cm}
 inLegalTree($u$), $(S.u \neq Power)$, and  
 (Child.$u \neq \emptyset$).
\newline
\hspace*{0.5cm}
A process $u$ verified the predicate unSafe iff
\newline
\hspace*{2cm}
 insideLegalTree($u$)  and
$(\exists v \in N(u) ~|~ C.v \neq C.u$ and $v$ is influential$)$.
\end{definition}

\noindent
The current execution may reach a configuration where  an unSafe process $u$ may verified the 
guard of the rule RC4 or RC5, and the predicate 
insideLegalTree($u$).

\begin{lemma}
\label{lem:unSafeClosed}
$\neg$unSafe($u$) is a closed predicate.
\end{lemma}
\begin{proof}
Let $cs$ be a computation step from $c1$ to $c2$.

\noindent
If $r$ executes the rule R1 or RC3 during $cs$ then no process verifies 
the predicate unSafe in $c2$ as no process verifies the Predicate 
insideLegalTree. 

\noindent
In the sequel, we assume that $r$ does not 
execute the rule R1 or RC3 during $cs$ (i.e.  $r$ does not change its  color). 

\medskip
\noindent
\textbf{Becoming influential.} Any process becoming influential during the step $cs$
as the $r$'s color in $c1$, and in $c2$ named $r\_color$ (Corollary~\ref{cor:preInf}).

\medskip
\noindent
\textbf{Becoming insideLegalTree.}
Assume that in $c1$, insideLegalTree($u$) is not verified but it 
is verified in $c2$. 
We conclude that $u$ has executed the rule R2 or R6 during $cs$;
so in $c1$, $u$ verifies  the predicates inLegalTree($u$) and  
EndFirstPhase($u$). Therefore in $c1$, $u$ and all $u$ neighbors 
have the same color as $r$, $r\_color$.  A neighbor of $u$, 
$v \neq r$,  would have not the color $r\_color$  in $c2$ 
if $v$ has executed the rule R3 during $cs$. 
In $c1$, $v$ would verify the predicate 
Connection($v$, $w$) but also the predicate StrongConflict($v$).
We conclude that during $cs$, $v$ could not execute the rule 
R3; so any neighbor of $u$, $v$, verifies 
$(C.v = r\_color = C.u)$ in $c2$. 
So  $\neg$unSafe($u$) 
is verified in $c2$.

\medskip
\noindent
\textbf{Becoming unsafe.}
Assume that in $c1$, unSafe($u$) is not verified but 
it is verified in $c2$. 
The  paragraph "becoming insideLegalTree" establishes that 
InsideLegalTree($u$) has to be verified in $c1$ to have 
unSafe($u$) in $c2$.
Let $v$ be a neighbor of $u$  that is influential in $c2$ 
and $C.v \neq C.u$, in $c2$. Such a process exists because 
unSafe($u$) is verified in $c2$.
We have $v \neq r$ as InsideLegalTree($u$) is verified in $c2$.
The paragraph "becoming influential" establishes that
if the process $v$ becomes 
influential during $cs$ then $C.v=C.u=r\_color$ at $c2$. 
So, $v$ is influential in $c1$. Hence,
$v$ cannot execute $R3$ during $cs$, so it cannot change its color.
Either $C.v \neq C.u$ in $c1$ and unSafe($u$) is verified in $c1$ or $C.v = C.u$ in $c2$ and unSafe($u$) is not verified in $c2$. \\
We conclude that no process becomes unSafe during $cs$.
\end{proof}

\begin{corollary}
\label{cor:unSafe-preInfl}
Let $cs$ be a computation step from $c1$ reaching $c2$. 
Let $u$ be a process that is not unRegular and influential in $c1$ 
but it is in $c2$. We have
$C.u=C.r$ in $c2$.
\end{corollary}
\begin{proof}
Let $u$ be a process that becomes unRegular and influential during $cs$.
Lemma~\ref{lem:becominginfluential3} establishes that if a process 
$u$ becomes influential during $cs$ then $u$ is not unRegular in $c2$.
So $u$ is influential and inLegalTree in $c1$. Therefore in $c1$, we have
$C.r=C.u=r\_color$.

\noindent
R1 is not enabled at $c1$ because inLegalTree($u$) is verified.
so $r$ has the same color $r\_color$ in $c1$ and in $c2$.
In $c1$, R3 is not enabled on $u$ because it is influential;
so $u$ has the same color, $r\_color$, in $c1$ and $c2$.
\end{proof}

%\begin{lemma}
%\label{lem:unSafe-preInflR0}
%Let $cs$ be a computation step from $c1$ reaching $c2$ 
%where no unSafe process  executes RC4, RC5, or RC3. 
%If $r$  executes  the rule RC3 during $cs$ then
%all unRegular and influential processes
%in $c2$ were also unRegular and influential processes in $c1$ 
%\end{lemma}
%
%\begin{proof}
%Assume that  $r$ executes the rule RC3 during $cs$. 
%So, in $c1$ we have $S.r = Power$ or $Child.r = \emptyset$
%as $\neg$unSafe($r$) is verified (by hypothesis).
%According to Lemma~\ref{lem:becominginfluential3}, no 
%process becomes influential during $cs$. 
%Thus, only processes verifying the both predicate 
%influential and inLegalTree in $c1$ may become influential and 
%unRegular during $cs$. Only $r$ verifies this both predicate in 
%$c1$; $r$ is not influential in $c2$ as $S.r = StrongE$.
%\end{proof}

\begin{lemma}
\label{lem:unSafe-preInfl}
Let $cs$ be a computation step from $c1$ reaching $c2$ 
where no unSafe process  executes RC4, RC5, or RC3. 
All unRegular and influential processes
in $c2$ were also unRegular and influential processes in $c1$.
\end{lemma}
\begin{proof}
Assume that $u_0$  becomes unRegular and  Influential during $cs$.

\medskip
\noindent
If $r$ is unRegular then $r$ is not influential as $S.r = StrongE$.
So $u_0$ is not $r$.
Lemma~\ref{lem:becominginfluential3} establishes that if a process 
$u_0 \neq r$ becomes influential during $cs$ then $u_0$ is not unRegular in $c2$.
So $u_0\neq r$ is influential and inLegalTree in $c1$. 
Therefore $P.u_0 \neq \perp$,
and $\neg$Faulty($u_0$) in $c1$.
According to Lemma~\ref{lem:FaultyIsaTrap}, 
we have $\neg$Faulty($u_0$) in $c2$. We name $u_1=P.u_0$ in $c1$. 
In $c1$, we have $S.u_1 \neq Power$ as $\neg$Faulty($u_0$) and 
influential($u_0$) are verified.

\medskip
\noindent
By induction, we will build an infinite series of distinct  processes $u_1$, $u_2$, 
$\cdots$ verifying
$u_i \neq r$, $P.u_i = u_{i+1}$, insideLegalTree($u_i$) 
for $i > 1$ in $c1$. 
As the number of processes is finite : we will establish a contradiction.

\medskip
\noindent
Either $u_0$ executes no rule during $cs$ or $u_0$ executes the rule R5 
(according to Lemma~\ref{lem:influentialMove}).
So $P.u_0=u_1$ in $c2$. We conclude that inLegalTree($u_1$) is true in $c1$ 
but not in $c2$. In $c1$, insideLegalTree($u_1$) is true  
(we have $S.u_1 \neq Power$ and $Child(u_1) \neq \emptyset$).
By hypothesis, $u_1$ does not execute the rule RC4, RC5 or RC3 during $cs$.

\medskip
\noindent
If $u_1 = r$ then $r$ would execute the rule RC3  during $cs$
because  inLegalTree($u_1$) is true in $c1$ 
but not in $c2$. 
We conclude that $u_1 \neq r$.
By definition of insideLegalTree($u_1$), we have  $P.u_1 \neq \perp$ (as $u_1 \neq r$)
and $\neg$Faulty($u_1$) in $c1$. Let us name $u_2=P.u_1$ in $c1$.
In $c1$, we have $S.u_2 \neq Power$ as $\neg$Faulty($u_1$) and 
$Child(u_1) \neq \emptyset$ are verified.
We also have  inLegalTree($u_2$) as  inLegalTree($u_1$) is verified.
%According to Lemma~\ref{lem:FaultyIsaTrap}, we have $\neg$Faulty($u_1$) in $c2$.
By hypothesis, $u_1$ does not execute the rule RC4, or RC5 
during $cs$. The rule R3 or R7 are disabled at $u_1$ in $c_1$.
So $P.u_1=u_2$ in $c2$. We conclude that inLegalTree($u_2$) is true in $c1$ 
but not in $c2$. 
In $c1$, insideLegalTree($u_2$) is true  
(we have $S.u_2 \neq Power$ and $Child(u_2) \neq \emptyset$).

\medskip
\noindent
By induction, we build an infinite series of  processes $u_1$, $u_2$, 
$\cdots$ verifying
$u_i \neq r$, $P.u_i = u_{i+1}$, insideLegalTree($u_i$) 
for $i > 1$ in $c1$. 
If it exists $i \neq j$ such that $u_i = u_j$ then 
insideLegalTree($u_i$) is not verified.
So all processes are distinct in the infinite series of processes: 
there is a contradiction.
\end{proof}

\section{Proof of execution Fairness}
\label{sec:fairnessProof}

\subsection{Inside-safe executions}

\begin{definition}[inside-safe execution]
An execution where insideLegalTree processes execute only the 
rules R1 to R7 and the number of  processes verifying the 
predicate Faulty stays unchanged is an inside-safe execution.
\end{definition}

\begin{lemma}
\label{lem:inside-safe}
Any execution has an inside-safe suffix.
The predicate ($\neg$influential
$\vee$ inLegalTree) is closed along an inside-safe suffix.
\end{lemma}
\begin{proof} 
Let $e$ be an execution.

\noindent
The predicate $\neg$Faulty($u$) is closed (Lemma
\ref{lem:FaultyIsaTrap}). So along $e$ the number of 
processes verifying the predicate Faulty can only decrease. 
We conclude that $e$ has a suffix where this number does not 
change.

\medskip
\noindent
The rules RC6, RC1 and RC2 are disabled on insideLegalTree processes.

\noindent
A  process verifying the predicate insideLegalTree cannot 
verified the predicates PowerFaulty, IllegalLiveRoot and
IllegalChild. 
Hence, an insideLegalTree process $u$ executes 
the rule RC4, RC5 or RC3 at $c1$ only if $u$ verifies 
StrongConflict or Conflict at $c1$ (i.e. $u$ is unSafe at $c1$).

\noindent
In the reached configuration $c2$, after a RC4, 
RC5 or RC3 move by $u$ we have
$P.u = \perp$ and $S.u \in Erroneous$;
so $u$ is not  unsafe in $c2$.
The predicate $\neg$Unsafe is closed (Lemma~\ref
{lem:unSafeClosed}). We conclude that any process $u$ executes at 
most one time the rule  RC4, RC5 or RC3  at a configuration where
insideLegalTree($u$) is verified.

\noindent
So $e$ has a inside-safe suffix, named $e'$.
Lemma~\ref{lem:unSafe-preInfl} establishes that the predicate
($\neg$influential
$\vee$ inLegalTree) is closed during any inside-safe execution.
\end{proof}

\noindent
The negation of the predicate
($\neg$influential($u$) $\vee$ inLegalTree($u$))
is the predicate (influential($u$) $\wedge$ unRegular($u$)).

\noindent
A process is silent  during an execution
if it does not execute any rule during this execution.

\begin{lemma}
Let $e1$ be an inside-safe execution $e1$ has a suffix 
where influential and unRegular processes are silent.
\end{lemma}
\begin{proof}
An influential process is not influential after any move 
except R5 (Lemma \ref{lem:influentialMove}).
According the R5 guard and action, the executed rule 
after a R5 move by a process cannot be R5. 

\noindent
So during $e1$, a process staying forever influential 
and unRegular executes at most a single rule (R5).
\end{proof}

\begin{lemma}
Let $e2$ be the inside-safe execution suffix where 
influential and unRegular processes are silent.
$e2$ has a suffix where no process executes RC4 and RC3.
\end{lemma}
\begin{proof}
Let $cs$ be a computation step of $e$ from $c1$ to $c2$.

\medskip
\noindent
Let $u$ be a process executing RC4 during $cs$.
In $c1$, two  processes of N[$u$], $v1$ and $v2$, 
have the $Power$ status and do not have the same color. 
So one of them, denoted in the sequel $v$,
verifies the predicate  (influential($v$) 
$\wedge$ unRegular($v$))  in $c1$. 
According to the definition of $e2$, 
$v$ keeps forever it status, so $u$  is 
disabled along $e$ from $c2$. 
We conclude that $e$ contains 
at most $n$ steps where a process executes RC4.

\medskip
\noindent
Assume that $r$ executes RC3 during $cs$.
{\em first case.} in $c1$, 
$r$ has a neighbor $v1$ having the Power status 
but not the $r$'s color.
{\em second case.} in $c1$, 
$r$ has a neighbor $v2$ having the Power status 
and $Child.r = \emptyset$ (i.e. no process except $r$ verifies the predicate inLegalTree). So
$v1$ or $v2$, named $v$ in the sequel, verifies the predicate  (influential($v$) 
$\wedge$ unRegular($v$))  in $c1$.
\noindent
According to the definition of $e2$, 
$v$ keeps forever its status, so $r$  is 
disabled along $e$ from $c2$. 
We conclude that $e$ contains 
at most one step where $r$ executes RC3.
\end{proof}

\begin{corollary}
\label{cor:PowerFaultyRoot}
Let $e3$ be the inside-safe execution suffix where 
influential and unRegular processes are silent, 
and no process executes RC4 and RC3.
$e3$ has suffix where the predicates $\neg$PowerFaulty and
$S.r \neq StrongE$ are closed.
\end{corollary}
\begin{proof}
Let $cs$ be a computation step of $e3$ from $c1$ to $c2$.

\medskip
\noindent
Assume that $S.r = StrongE$ in $c2$.
During $e3$, no process takes the status $StrongE$.
So in $c1$, we have $S.r = StrongE$. 
The predicate $S.r \neq StrongE$ is closed along $e3$.

\medskip
\noindent
Assume that PowerFaulty($u$) is verified  in $c2$.
During $e3$, no process takes the status $StrongE$.
So, $u$ has a neighbor, $v$,  having the $StrongE$ status 
in $c1$.
$u$ cannot execute the rule R1 or R5 to take the $Power$ 
status during $cs$ according to the rule guards.
We conclude that in $c1$, $S.u = Power$: PowerFaulty($u$) 
is verifies in $c1$.
\end{proof}

\begin{lemma}
Let $e3$ be the inside-safe execution suffix where influential 
and unRegular processes are silent, no process executes 
RC4 and RC3, and the predicates $\neg$PowerFaulty and
$S.r \neq StrongE$ are closed.
$e3$ has suffix where $r$ does not execute the rule RC1 and RC2.
\end{lemma}
\begin{proof}
Let $cs$ be a step of $e3$ from $c1$ to reach $c2$ where
$r$ executes a rule. 
If $r$ executes RC1 or RC2 then the predicate PowerFaulty($r$) or
the predicate $S.r = StrongE$ is verified in $c1$.
In the sequel, we will prove that these both predicates are no verified after any $r$ move.

\medskip
\noindent
If $r$ executes the rule RC1, R2, or RC2 then in $c2$, we have
$S.r = Working$. So,
$\neg$PowerFaulty($r$) and $S.r \neq StrongE$ are verified 
in $c2$;
these both predicates stay verified along the execution 
$e3$ from $c2$ (Corollary~\ref{cor:PowerFaultyRoot}). 
So RC1 and RC2 are not executed along $e3$ 
from $c2$.

\medskip
\noindent
If $r$ executes the rule R1 then in $c1$, $\neg$PowerFaulty($r$) 
and $S.r \neq StrongE$ are verified. 
These both 
predicates stay verified along the execution 
$e3$ from $c2$ (Corollary~\ref{cor:PowerFaultyRoot}).
So RC1 and RC2 are not executed along $e3$ 
from $c2$.

\noindent
We conclude that $r$ executes at most one time the rule 
RC1 and the rule RC2 along $e3$.
\end{proof}

\subsection{Safe executions}

\begin{definition}[safe execution]
\label{def:safe-execution}
A safe execution is an execution where insideLegalTree processes 
execute only the rules R1 to R7, influential and unRegular 
processes are silent, the predicate $\neg$PowerFaulty is closed 
and  the rule RC1, RC4, RC3 and RC2 are never executed.
\end{definition}

A safe execution is an inside-safe execution.

\begin{corollary}
\label{cor:safe}
Any execution has a safe suffix.
\end{corollary}

\begin{lemma}
\label{lem:ConflitRC5}
Let $e$ be a safe execution 
$e$ has a suffix where no process executes the rule RC5 at a 
configuration where it verifies Conflict and inLegalTree.
\end{lemma}
\begin{proof}
Assume that a  process $u \neq r$ 
executes the rule RC5 during the step $cs$ from $c1$ to $c2$ along $e$;
moreover, in $c1$ we have  inLegalTree($u$).
So $u$ has a neighbor $v$ verifying $S.v = Power$ 
and $C.v = \overline{r\_color}$ ($r\_color$ being the color of $r$ in $c1$).
According to the definition of a safe-execution, $v$ is silent 
along $e$ because $v$ verifies the predicates influential and 
unRegular. 
After $u$ move, $u$ verifies $P.u = \perp$. 
Assume that along $e$ from $c2$ there is a computation step from $c3$ to $c4$ 
where $(C.u = \overline{r\_color}) \vee (P.u = \perp)$ is verified in $c3$ but not in $c4$. 
During $cs$, $u$ must have execute the rule R3 to choose as a parent a process having 
the color  $r\_color$.
In $c3$, $u$ has a neighbor $w$ 
verifying  $C.w = r\_color$  
and $S.w = Power$, $u$ verifies the predicate 
StrongConflict not  the R3 guard (because  $S.v = Power$ and $C.v = \overline{r\_color}$).
So $u$ never executes the rule R3 to set its color to $r\_color$.
We conclude that $u$ verifies forever 
$(C.u = \overline{r\_color}) \vee (P.u =\perp)$.

Assumes that $u$ verifies the predicate Conflict and inLegalTree
in the configuration $c3'$ reached along $e$ from $c2$.
We have established that $C.u = \overline{r\_color}$ in $c3'$.
So $u$ has a neighbor $w$ verifying  $C.w = r\_color$  
and $S.w = Power$; we conclude that $u$ does not verify the RC5 guard but the RC4 guard
because  $S.v = Power$ and $C.v = \overline{r\_color}$ in $c_3'$. 

We conclude that $u$ executes at most one time the rule RC5 at a 
configuration where it verifies Conflict and inLegalTree along $e$.
\end{proof}

\begin{lemma}
\label{lem:FaultyRC5}
Let $e$ be a safe execution 
$e$ has a suffix where no process executes the rule RC5 at a 
configuration where it verifies PowerFaulty or Faulty.
\end{lemma}
\begin{proof}
Let $cs$ be a step of $e$ from $c1$.
Let $u$ be process $u$ that executes the rule RC5 during $cs$.

\medskip
\noindent
Assume that $u$ verifies the predicate Faulty or PowerFaulty in 
$c1$. After the step $cs$, $u$ never verifies
the predicate Faulty and the predicate PowerFaulty according to 
the definition of a safe execution and Lemma 
\ref{lem:FaultyIsaTrap}.
So it never executes again RC5 at a configuration where these 
predicates are verified.
\end{proof}

\begin{lemma}
Let $e$ be a safe execution 
$e$ has a suffix where no process executes the rule RC5 at a 
configuration where it verifies  inLegalTree.
\end{lemma}
\begin{proof}
A process $u$ verifying the predicate inLegalTree does not 
verify the predicate IllegalRoot and does not verify 
the predicate IllegalChild.
$e$ has a suffix $e1$ where no process verifying 
PowerFaulty or Faulty executes RC5 (Lemma~\ref{lem:FaultyRC5}).
$e1$ has a suffix $e2$ where RC5 is not executed by processes 
verifying inLegalTree and Conflict (Lemma~\ref{lem:ConflitRC5}).
We conclude that in $e2$ no process verifying inLegalTree 
executes RC5.
\end{proof}

\begin{lemma}
Let $e1$ be a safe execution where no inLegalTree process 
executes the rule RC5.
$e1$ has a suffix where the predicate 
$\neg$unRegular is closed.
\end{lemma}
\begin{proof}
Along $e1$, a process $u$ becomes unRegular 
if during the execution of the 
rule R3: $u$ chooses as parent an unRegular process $v$ having 
the status $Power$.
According to the definition of safe-execution, $v$ is silent 
along $e$.  
So, the only enabled rule at $u$ after this move is RC4. 
We conclude that the process $u$ is silent after its R3 action 
along $e1$. As a process may become unRegular at most one time along $e1$; $e1$ has a suffix where the predicate $\neg$unRegular is closed.
\end{proof}

\begin{definition}[pseudo-regular executions]
\label{def:pseudo-regular}
A pseudo regular execution is a safe execution 
where  inLegalTree processes execute only the rules R1 to R7, 
and the predicate $\neg$unRegular($u$) is closed.
\end{definition}

In  the section \ref{sec:convergence}, we established that any execution 
from a configuration of {\em A4} is pseudo-regular.
{\em A4} is an attractor reached in $16n-13$ rounds.

\begin{lemma}
Let $e2$ be a pseudo-regular execution.
$e2$ has a suffix where no process  executes rule RC5.
\end{lemma}
\begin{proof}
Let $u$ be a process executing RC5 along $e2$ from $c1$ to reach 
the configuration $c2$. In $c1$, unRegular($u$) is verified.
\noindent
In $c2$, $S.u \in Erroneous$ and $P.u = \perp$.
The next move of $u$, 
if it exists, is the execution of the rule RC6 or R3.
Assume that $u$ executes RC6 or R3 at the configuration $c3$. 
Detached($u$) is verified in $c3$. So
the predicate $\neg$unRegular($u$) is verified along $e2$ 
from $c3$ (by definition of $e2$).

\noindent
Along $e2$, a process $u$ executes at most one time the rule RC5.
\end{proof}

\noindent
Let $e3$ be a a pseudo-regular execution where 
no process executes the rule RC5.
$S.u \notin Erroneous$ is closed along $e3$.  So a process 
executes at most one time the rule RC6 along this execution.

\subsection{Regular executions}
\begin{definition}[regular execution]
A regular execution is a pseudo-regular execution 
where  processes execute only the rules R1 to R7, 
and the predicate $\neg$unRegular($u$) is closed.
\end{definition}

\begin{corollary}
\label{cor:regular}
Any execution has a regular suffix.
\end{corollary}

\noindent
During a regular execution, the move of a process $u \neq r$ 
belongs to the language: (R3 R5 (R6R4)$^*$ R7)$^*$, and
 the moves of $r$ belong to the 
language : (R1 R2$^*$)$^*$.

\begin{lemma}
Let $e$ be a regular execution where 
a process $u$ never changes its color.
The execution $e$ has a suffix where all processes of N($u$) do not 
change their color.
\end{lemma}
\begin{proof}
Assume that $u$ keep the color $co$ along $e$.
A process changes its color by the execution of 
the rule R3 or the rule R1.
Let $v$ be a process of N($u$).
$v$ never verifies EndFirstPhase($v$) 
if $C.v = \overline{co}$ along $e$.
So the execution of the rule R3 or R1 by $v$ 
to take the color $C.v = \overline{col}$ 
is the second to last or the last move of $v$.
Therefore, $v$ changes its color at most 2 times along $e$ : 
one to take the color $co$, and the last one to take the color 
$\overline{co}$.
\end{proof}

\begin{corollary}
Let $e$ be  a regular execution where 
a process $u$ never changes its color.
The execution $e$ has a suffix where no process changes its color.
\end{corollary}

\noindent
In the sequel of this section, we study regular executions where 
no process changes its color. We will establish that such 
executions do not exist.

\begin{observation}
Let $exc$ be a regular execution where no process change its 
color.
\noindent
There is no terminal configuration; so 	at least a process 
along $exc$ executes infinitely often a rule.

\noindent
The moves of any process $u \neq r$ belong 
to the language R5R7 or to the language (R6R4)$^*$ R7. 
The moves of $r$ belong to the language R2$^*$.

\noindent
Along $exc$ no process joins any Child($u$) set.
\end{observation}

\begin{lemma}
Let $exc$ be a regular execution where where no process 
changes its color.
Let $u$ be a process executing  infinitely often a rule 
along $exc$.
$u$ has  forever a child that also executes infinitely 
often a rule.
\end{lemma}
\begin{proof}
The moves of $u$ belong to the language $(R6R4)^*$ 
if $u \neq r$ or to $R2^*$. So $u$ changes its phase value 
infinitely often  (rule R4 or R2). 
$u$ cannot gain a child. From a configuration where  
the rule R4 or R2 are enabled at $u$, Child($u$) 
is not empty and QuietSubTree($u$) is verified. 
So at least a process $v$ verifies forever $P.v=u$ and 
$v$ changes its phase value infinitely often along $exc$.
\end{proof}

\begin{corollary}
\label{cor:reg-cycle}
Let $exc$ be a regular execution where no process 
changes its color.
There a cycle path whose all the processes execute infinitely 
often a rule.
\end{corollary}

\begin{theorem}
Let $exc$ be a regular execution.
Along $exc$ every process changes its color infinitely often.
\end{theorem}
\begin{proof}
Let $exc$ be a regular execution. Assume that no process 
changes its color along $exc$. 

\noindent
There a cycle path $u_1$, $u_2$, ... $u_l$ whose all 
the processes execute infinitely often a rule along $exc$ 
(Corollary~\ref{cor:reg-cycle}).
As any process $u_i$ is not silent along $exc$; we have
$S.u_i \in \{ Working, Idle \}$ along $exc$.
In a configuration reached by $exc$, if $S.u_i = Idle$ then 
$S.u_{i+1} = Idle$ otherwise Faulty($u_{i+1}$) would be 
verified (and $u_i$ and $u_{i+1}$ would be silent).
So in any configuration reached by $exc$, 
we have  $S.u_i = S.u_{i+1}$ with $0 < i < l$. 

\noindent
There is a contradiction as no process of the cycle path is 
enabled at a configuration where $S.u_i = S.u_{i+1}$ 
for all $0 < i < l$.
\end{proof}

\noindent
Every execution of the algorithm has a regular suffix.
During a regular execution, every process changes its color 
infinitely often (i.e. they execute R3 or R1 rule infinitely 
often). At a configuration where R1 or R3 is enabled at $u$ we have 
$\neg$unRegular($u$).

%\begin{corollary}
%\label{cor:reg-cycleAtt}
%Every execution of the algorithm is fair under any scheduler.
%A1 is the set of configurations defined as 
%\{ $\forall$ $u \in V$, $S.u \notin Erroneous$ and 
%$\neg$unRegular($u$) \}.
%A1 is A0-attractor.
%\end{corollary}

\section{Convergence Time}
\label{sec:convergence}

\begin{lemma}
The predicate $\neg$IllegalLiveRoot($u$) is closed.
\label {lem:IllegalRootisTrap}
\end{lemma}
\begin{proof}
There is not creation a new illegal root that does not have 
an $Erroneous$ status:
\begin{list}{$\bullet$}{\setlength{\partopsep}{0.00in} 
\setlength{\topsep}{0.00in} 
\setlength{\itemsep}{0.04in}
\setlength{\parsep}{0.00in}}
\item An  process that does not have the $Power$
status cannot gain children.
\item An process $u  \neq r$ without parent 
cannot take the $Power$  status.
\item At time a process sets $P$ to $\perp$ it also sets 
its status to $StrongE$ or $WeakE$ if it has or may 
gain children.
\end{list}
\end{proof} 

\begin{theorem}
Let {\em A0} be the set of configurations.
{\em A1} $\equiv$ \{ $\forall$ $u$, $\neg$Faulty($u$) 
$\wedge$ $\neg$IllegalLiveRoot($u$) \}.
is an {\em A0}-attractor reached from $A0$ in one 
round.
\label{A1attractor}
\end{theorem}

\begin{proof} 
$\neg$Faulty($u$) is closed 
(Lemma \ref{lem:FaultyIsaTrap}) 
and $\neg$IllegalLiveRoot($u$) is closed 
(Lemma \ref{lem:IllegalRootisTrap}).
So {\em A1} is closed.

\noindent
As long as a process $u$ satisfies the Faulty 
or IllegalLiveRoot predicate, $u$ is enabled (it verifies 
the RC4 or RC5 guard). After  RC4 or RC5 action by $u$, we have 
$\neg$Faulty($u$) 
$\wedge$ $\neg$IllegalLiveRoot($u$). So, 
after the first round, {\em A1} is reached.
\end{proof} 

\subsection{Erroneous processes}

\begin{lemma}
\label{lem:detached}
Let $c$ be a configuration where the process $u$ has  an
$Erroneous$ status. Let $e$ be a execution from $c$ 
where $u$ keeps its status more than one round.
At the end of the first round of $e$ from $c$, Detached($u$) 
is verified and this property stays verified 
till $u$ has an $Erroneous$ status.
\end{lemma}
\begin{proof}
Till $S.u \in Erroneous$, no process chooses $u$ as Parent.
Let $v$ be a process of Child($u$) in $c$. We have $v\neq r$.
The rule RC4 or RC5 is enabled at $v$ till $S.u \in Erroneous$. 
So $v$ quits Child($u$) before the end of the first round from $c$.
We have Child($v$) $= \emptyset$ at then end of the first round
from $c$.

\noindent
If $P.u \neq \perp$ (we have $u \neq r$) then RC4 or RC5 is enabled 
at $v$. So Detached($u$) is verified at the end of 
the first round from $c$. To conclude, we notice that the predicate
Detached($u$) $\vee (S.u \notin Erroneous)$ is closed.
\end{proof}

\begin{lemma}
\label{lem:StrongE}
Let $c$ be a configuration of {\em A0} where the process $S.u = StrongE$ status. 
Let $e$ be a execution from $c$. Before the end of the second round
$u$ has changed its status to $Idle$.
\end{lemma}
\begin{proof}
Assume that $u$ keeps its status during the first round of $e$.

\medskip
\noindent
The predicate Detached($u$) is verified at then end 
of the first round from $c$.
The predicate Detached($u$) stays verified till $S.u = StrongE$ 
(Lemma \ref{lem:detached}).

\medskip
\noindent
Till $S.u = StrongE$, no process in N($u$) takes the $Power$ status.
Let $v$ be a process of N($u$) verifying $S.v = Power$ in $c$. 
The rule RC1, RC4, RC5, RC3 is enabled at $v$ till $S.u = StrongE$. 
So $S.v \neq Power$ before the end of the first round from $c$.
At the end of the first round, we have StrongEReady($v$); this property stays verified
 till $S.u = StrongE$.

\medskip
\noindent
During the second round, $u$ is always enabled (R3, RC6, or RC2). 
So $u$ executes one of this rule that changes 
it status to $Idle$.
\end{proof}

\begin{lemma}
\label{lem:WeakE}
Let $c$ be a configuration where the process $u$ has the  status
$WeakE$. 
Let $e$ be a execution from $c$. Before the end of the second round,
$u$ has changed its status to $Idle$ or to $StrongE$.
\end{lemma}
\begin{proof}
We have $u \neq r$. 
Assume that $u$ keeps its status during the first round of $e$.
The predicate Detached($u$) is verified at the end of 
the first round from $c$ and the predicate
Detached($u$) stays verified till $S.u = WeakE$ 
(Lemma \ref{lem:detached}).

\noindent
Till $S.u = WeakE$ and Detached($u$), 
one of the rule RC4, RC6 or R3 is enabled on $u$. 
So before the end of the second round, 
$u$ has changed its status
to $Idle$ or $StrongE$.
\end{proof}

%\colette{$S.u = WeakE$ et Detached($u$) est vérifié. 
%Il y a une question philosophique. 
%$u$ est toujours executable durant la seconde rounde 
%mais la regle applicable change en fonctionde 
%l'état de son voisinage.}

\subsection{Power processes}

\begin{lemma}
\label{lem:PoweR15}
Let $c$ be a configuration of {\em A0} where a process $u$ has the
$Power$ status. Let $e$ be an execution from $c$.
Let $v$ be a $u$ neighbor. The predicate
$C.v = C.u$ holds along $e$ till $u$ keeps the $Power$ status.
\end{lemma}
\begin{proof}
On a configuration where $S.u = Power$ and $(C.v = C.u)$, 
R1 and R3 are not enabled at $v$. 
So till $S.u = Power$ the predicate $(C.v = C.u)$ is closed.
\end{proof}
\begin{lemma}
\label{lem:Power}
Let $c$ be a configuration of {\em A0} 
where the process $u$ has the 
$Power$ status. Let $e$ be an execution from $c$ where $u$ keeps 
its status $Power$ during at least a round.
At the end of the first  round of $e$ from $c$,  
we have 
$\{\forall v \in N(u) ~|~ (C.v = C.u) ~\vee~ (S.v \in Erroneous)\}$, 
this property holds along $e$ till $u$ keeps the $Power$ status.
\end{lemma}
\begin{proof}
Let $e$ be a execution from $c$ where $u$ keeps 
the status $Power$ during a round starting at the configuration $c$.
Let $v$ be a process of $N(u)$ verifying 
$(C.v \neq C.u) \wedge (S.v \notin Erroneous)$  at $c$.
Till $\neg$Ok($v$) is verified, RC4, RC5, RC3 or RC1 is enabled at $v$.
Till both predicates Ok($v$) and $(C.v \neq C.u) \wedge (S.v \notin Erroneous)$ 
are verified during this round, Connection($v$, $u$) is verified. So,
R3 is enabled at $v$ if $v \neq r$.
$v=r$ cannot verify the predicates Ok($r$) 
and $(C.r \neq C.u) \wedge (S.r \notin Erroneous)$. 

\noindent
We conclude that till $(C.v \neq C.u) \wedge (S.v \notin Erroneous)$ is verified, 
$v$ is enabled (rule R3, RC5, RC4, RC3, or RC1). 
After the $v$ move, we have $(C.v = C.u) ~\vee~ (S.v \in Erroneous)$.

\medskip
\noindent
Let $v$ be a process of $N(u)$ verifying 
$(C.v = C.u) ~\vee~ (S.v \in Erroneous)$  at a configuration $c$, where $S.u = Power$.
Assume that $(C.v \neq C.u)$ in $c$, we have 
$(S.v \in Erroneous)$ in $c$ by hypothesis. 
Only the rule R3, RC4, RC5 may be  enabled at $v$.
We have $(C.v = C.u) ~\vee~ (S.v \in Erroneous)$
in  reached configuration after a move by $v$.
Till $S.u = Power$ the predicate
$(C.v = C.u)$ is closed (Lemma~\ref{lem:PoweR15}).
\end{proof}

\begin{observation}
\label{obs:PowerQuiet}
Let $c$ be a configuration of {\em A1}.
Any $Power$ process $u$ verifies QuietSubTree($u$) 
otherwise $u$ would have a Faulty child.
Any process $v$ verifying the predicate NewPhase verifies QuietSubTree($v$) 
otherwise $v$ would have a Faulty child.
\end{observation}
\begin{theorem}
\label{theo:power}
Let $c$ be a configuration of {\em A1} 
where the process $u$ has an 
$Power$ status. Let $e$ be a execution from $c$.
Before the end of the 4th  round of $e$ from $c$, 
$u$ has changed is status or $u$ has executed the rule R1.
\end{theorem}
\begin{proof}
If $u$ has a neighbor having the status $StrongE$ 
before or at the end of the third round then PowerFaulty($u$) is verified 
till $S.u = Power$. So $u$ changes its status before the end 
of 4th round.
In the sequel, we study the executions $e$ from $c$ 
where no $u$ neighbor has the status $StrongE$ 
during the first $3$ rounds of $e$.

\medskip
\noindent
Assume that $u$ keeps 
its status $Power$ during the first $3$ rounds of $e$ from $c$.
At the end of the first round,  we have 
$\{\forall v \in N(u) ~|~ (C.v = C.u) ~\vee~ (S.v \in Erroneous)\}$, 
this property holds along $e$ till $u$ keeps the $Power$ status
(Lemma \ref{lem:Power}).
Let $v1$ be  a neighbor of $u$,  verifies $C.v1 = C.u$ in $c$
Till $S.u = Power$ the predicate
$(C.v1 = C.u)$ is closed (Lemma~\ref{lem:PoweR15}).

\medskip
\noindent
Let $v2$ be a neighbor of $u$ having the status $WeakE$ at the end 
of the first round of $e$ from $c$.
Before the end of the third round,
$v2$ has changed its status to $Idle$ or $StrongE$ 
(Lemma \ref{lem:WeakE}).
As no $u$ neighbor has the $StrongE$ status (by hypothesis), 
before or at the end of the third round from $c$ of $e$,
a configuration $c_ {v2}$ where $C.v2 = C.u$ is reached.
Till $S.u = Power$ the predicate
$(C.v2 = C.u)$ is closed (Lemma~\ref{lem:PoweR15}).
Therefore, before or at the end of the third round 
from $c$ of $e$,
a configuration $c'$ where
$\{\forall v \in N(u) ~we~have~ (C.v = C.u) \}$ is reached.
From $c'$, till  $S.u = Power$, we have 
 QuietSubTree($u$) 
(Observation \ref{obs:PowerQuiet}).
So, from $c'$ till  $S.u = Power$, $u$ is enabled  
(RC5, RC4, RC3, R7, R6, R2, or R1).
We conclude that $u$ changes its status 
before the end of 4th round or $u$ has  executed the rule R1.
\end{proof}

\subsection{Unsafe processes}

\begin{definition}
Let {\em A2} be the set of configurations defines as
{\em A2} $\equiv$ 
{\em A1} $\cap\, $\{$\forall$ $u \in V$, $\neg$Unsafe($u$)\}.
\end{definition}

\noindent
In the sequel, we establish that in at most $8n-8$ rounds from a 
configuration of {\em A1}, {\em A2} is reached along any 
execution; and {\em A2} is a {\em A1}-attractor.

\begin{definition}
Let $c$ be a configuration of {\em A1}. 

The predicate PIC($u$) on process is defined as 
PIC($u$) $\equiv$ influential($u$) $\wedge$ $(C.u \neq C.r)$.
$\#PIC_{c} =|\{ u ~|~ PIC(u)~is~verified~in~c.~\}|$.

The predicate PIC\_PowerParent($u$) on process is defined as \\ \hspace*{2cm}
PIC\_PowerParent($u$) $\equiv$ 
PIC($u$) $\wedge$ PowerParent($u$).

$\#PIC\_PowerParent_{c} =|\{ u ~|~ PIC\_PowerParent(u)~is~verified~in~c.~\}|$.
\end{definition}

\begin{lemma}
\label{lem:PICPowerParentClosure}
Let $c_1$ be a configuration of {\em A1}. 
Let $cs$ be a step from $c1$ to $c_2$ where $r$ does not change its color.
If the predicate PIC($u$) (resp. PIC\_PowerParent($u$))
is verified in $c_2$ then it is verified in $c_1$ 
(resp. PIC\_PowerParent($u$)).
\end{lemma}
\begin{proof}
If $u$ becomes influential during $cs$ then 
$C.u=C.r$ in $c_2$ 
(Corollary~\ref{cor:preInf}, page \pageref{cor:preInf}). 
If $u \neq r$ changes of color during $cs$ (i.e. $v$ executes  R3) 
then $u$ is not influential in $c_2$.  
We conclude that if PIC($u$) is verified in $c_2$ 
then it is verified in $c_1$ as $r$ does not change its color during $cs$.

If in $c_1$, $S.u = Power$ is verified 
then the predicate PowerParent($u$) is not verified in $c_2$.
So, if PIC\_PowerParent($u$) is verified in $c_2$ 
then it is verified in $c_1$.

So $\#PIC_{c_2} \leq PIC_{c_1}$ 
and  $\#PIR\_PowerParent_{c_2} \leq PIR\_PowerParent_{c_1}$.
\end{proof}

\begin{lemma}
\label{lem:PICPowerParentConvergence2}
Let $c_1$ be a configuration of {\em A1} 
where 
$\#PIC_{c_1}  > 0$.
Let $e$ be a execution from $c_1$.
Let $c_4$ be the configuration reached after 4 rounds along $e$.
Assume that during the first $4$ rounds of $e$, $r$ does change its color.
If 
$\#PIC_{c_4} = \#PIC_{c_1}$ then
%There is at least a process $u_1$ that verifies the predicate
%$PIC\_PowerParent(u_1)$ in $c_1$ but not 
%at the end of the 4th  round of $e$.
$\#PIC\_PowerParent_{c_4} < \#PIC\_PowerParent_{c_1}$.
\end{lemma}
\begin{proof}
In the proof, we assume that $\#PIC_{c_4} \geq \#PIC_{c_1}$.

\medskip
Let $cs$ be a computation step from $c$ to reach $c'$ belonging 
to the first $4$ rounds of $e$.
According to Lemma~\ref{lem:PICPowerParentClosure}, 
if in $c'$, we have PIC($u$) then PIC($u$) 
is verified in $c$ and in $c_1$.
According to the the hypothesis, 
any process verifying PIC in $c_1$,
verifies this predicate in any configuration 
reached during the first $4$ rounds of $e$. % :  $\#PIC_{c_4} = \#PIC_{c_1}$. 
According to Lemma~\ref{lem:influentialMove} (page~\pageref{lem:influentialMove}),
the only rule that may execute a process $u^1$ verifying PIC in $c_1$ during the 
the first $4$ rounds of $e$ is R5.
Assume that $u^1$ executes R5 during $cs$.  
In $c$, PowerParent($u^1$) is verified but not in $c'$.
According to Lemma~\ref{lem:PICPowerParentClosure},
we have $\#PIC\_PowerParent_{c_4} < \#PIC\_PowerParent_{c_1}$.

\medskip
In the sequel, we study executions $e$ 
where processes verifying PIC in $c_1$ are silent during the  first $4$ rounds.

Assume in $c_1$, there is a process $u^2$ verifying 
$PIC(u^2) ~\wedge~ S.u^2 =Power$ : 
we have $u^2 \neq r$.
According to Theorem~ \ref{theo:power}, $u^2$ executes a rule during 
the first $4$ rounds of $e$.
There is a contraction. We conclude that in $c_1$, we have 
$(PIC(u) \implies PIC\_PowerParent(u))$.
As, processes verifying PIC in $c_1$ are silent during the first $4$ rounds of $e$;
in any reached configuration we have $(PIC(u) \implies PIC\_PowerParent(u))$.
So during the first $4$ rounds, no process verifies PowerConflict and 
Conflict($r$) is not verified as any process $u$ 
verifies $S.u \neq Power ~\vee~ C.u = C.r$. 
Hence, the rule RC3 and RC4 are not executed during the first $4$ rounds of $e$. 

Let $cs$ be a computation step from $c$ to $c'$ 
belonging to the first $4$ rounds of $e$.
If $S.w = StrongE$ in $c'$ then $S.w = StrongE$ in $c$ and in $c_1$.
Let $c_2$ be the configuration reached at the end of 
the 2th round of $e$ from $c$. 
According to the Lemma~\ref{lem:StrongE} (page~\pageref{lem:StrongE}),
in $c_2$ and during the  first round
from $c_2$ along $e$, every process verifies $S.w \neq StrongE$.

A process $u$ verifying PowerParent($u$) as the $Idle$ status
so QuietSubTree($u$) is verified as no process is faulty
in {\em A1}.
According to the definition of the predicate PowerParent,
in $c_2$ there is a process $u^2$ verifying PIC\_PowerParent($u^2$) and  
$ph.P.u^2 \neq ph.u^2$. Let $v$ be $P.u^2$ in $c_2$.
Till PowerParent($u^2$) and $ph.P.u^2 \neq ph.u^2$ are verified 
the rules RC1, RC2, RC6, R1, R2, R3, R4, R5, 
R6 and R7 are disabled on $v$.
If $v$ executes the rule RC3, RC4, RC5 then in the reached configuration, 
we have $S.v \in Erronneous$ so PowerParent($u^2$) is not verified.
Hence, according to Lemma~\ref{lem:PICPowerParentClosure}, 
we have $\#PIC\_PowerParent_{c_4} < \#PIC\_PowerParent_{c_1}$.

So, $v$ does not execute any rule from $c_2$ along $e$ during
first $2$ rounds.
Therefore, NewPhase($u^2$) is verified in $c_2$ and
during the 3th round of $e$ from $c_1$.
Till NewPhase($u^2$) is verified, $u^2$ may execute
a rule (RC4, RC5, R4, or R5) because every process verifies $S.w \neq StrongE$.
The process $u^2$ cannot be silent during the third round of $e$.
We conclude that at least a process verifying PIC in $c_1$, 
does a move during the first $4$ rounds.
So, we have $\#PIC\_PowerParent_{c_4} < \#PIC\_PowerParent_{c_1}$.
\end{proof}

\begin{lemma}
\label{lem:A2-attractor}
{\em A2} is a {\em A1}-attractor reached in at most $8n-8$ rounds 
from a configuration of {\em A1} 
(i.e. {\em A2} is a {\em A0}-attractor reached in at most $8n-7$ rounds 
from any configuration).
\end{lemma}
\begin{proof}
For any process $u$, $\neg$unSafe($u$) is a closed predicate 
(Lemma~\ref{lem:unSafeClosed}, page~\pageref{lem:unSafeClosed}).

\medskip
\noindent
Let $c1$ be a configuration of {\em A1}. 
Let $e$ be an execution starting at $c1$. 
After the execution of R1, no process is Unsafe as no process 
verifies insideLegalTree.

\noindent
In the sequel, we assume that during the first $8n-8$ rounds of $e$, 
the rule R1 is not executed.
So $r$ does not change its color $r\_color$ during the first $8n-8$  
rounds of $e$.

\noindent
In a configuration where  no  process $u$ verifying 
PIC($u$),
no process verifies the predicate unSafe
 (notice that $\neg$unSafe is a closed predicate).
The predicates $\neg$PIC and  $\neg$PIC\_PowerParent are closed  
along  the first $8n-8$ rounds of $e$ (Lemma~\ref{lem:PICPowerParentClosure}).
%In the sequel of the proof, 
%we establish that  
%after $4n$ rounds no process verifies the predicate PIC.

\medskip
\noindent
\textbf{convergence part.}\\
Let $c$ be a configuration of {\em A1}.
$ValPIC_c = \#PIC\_PowerParent_c + \#PIC_c$.
We have $ValPIC_c \leq 2n-2$ as PIC($r$) is not verified.
Let $cs$ be a computation step from $c$ to $c'$ done 
during the first $8n-8$ rounds of $e$.
As the predicates $\neg$PIC and  $\neg$PIC\_PowerParent are closed  
along  the first $8n-8$  rounds of $e$ 
(Lemma~\ref{lem:PICPowerParentClosure}); 
we have $ValPIC_c \geq ValPIC_{c'}$.

\noindent
Let $c4$ be the configuration reached the end of first $4$ rounds of $e$ from $c1$.
According to Lemma~\ref{lem:PICPowerParentConvergence2},
we have $Val_{c4} < ValPIC_{c1}$ if $ValPIC_{c1} > 0$.\\
Let $cf$ be the configuration reached the end of first $8n-8$ rounds of $e$ from $c1$.
So, we have $ValPIC_{cf} = 0$ (i.e. no  process $u$ verifying PIC($u$)).
We conclude that no process is unSafe at the end of first $8n-8$ rounds of $e$ from $c1$.
\end{proof}

\subsection{Influential and unRegular processes}

\begin{definition}
Let {\em A3} be the set of configurations defined as 
{\em A3} $\equiv$ 
{\em A2} $\cap$ \{ $\forall$ $u$, $\neg$Influential($u$) 
$\vee$ inLegalTree($u$)\}.
\end{definition}

\noindent
In the sequel, we establish that in at most $8n-8$ rounds from a 
configuration of {\em A2}, {\em A3} is reached along any 
execution; and {\em A3} is a {\em A2}-attractor.

\begin{definition}
Let $c$ be a configuration of {\em A1}. 

The predicate PIR($u$) on process is defined as 
PIR($u$) $\equiv$ influential($u$) $\wedge$ unRegular($u$).
$\#PIR_{c} =|\{ u ~|~ PIR(u)~is~verified~in~c.~\}|$.

The predicate PIR\_PowerParent($u$) on process is defined as \\ \hspace*{2cm}
PIR\_PowerParent($u$) $\equiv$ 
PIR($u$) $\wedge$ PowerParent($u$).

$\#PIR\_PowerParent_{c} =|\{ u ~|~ PIR\_PowerParent(u)~is~verified~in~c.~\}|$.
\end{definition}

\begin{observation}
The negation of the predicate
($\neg$influential($u$) $\vee$ inLegalTree($u$))
is the predicate PIR($u$). 
\end{observation}

\begin{lemma}
\label{lem:PIRPowerParentClosure}
Let $c_2$ be a configuration of {\em A2}. 
Let $cs$ be a step from $c_2$ to $c_3$.
If the predicate PIR($u$) (resp. PIR\_PowerParent($u$))
is verified in $c_3$ then it is verified in $c_2$ 
(resp. PIR\_PowerParent($u$)).
\end{lemma}
\begin{proof}
%Let $cs$ be a computation step from $c_2$ a configuration 
%of {\em A2} reaching $c_3$.
According to lemma \ref{lem:unSafe-preInfl} 
(page~\pageref{lem:unSafe-preInfl}) and {\em A2} definition, 
if a process $u$ verifies PIR($u$) 
in $c_3$ is also verifies this predicate in $c_2$.

If in $c_2$, $S.u = Power$ is verified 
then the predicate PowerParent($u$) is not verified in $c_3$.
So, if PIR\_PowerParent($u$) is verified in $c_3$ 
then it is verified in $c_2$.

So $\#PIR_{c_3} \leq PIR_{c_2}$ 
and  $\#PIR\_PowerParent_{c_3} \leq PIR\_PowerParent_{c_2}$.
\end{proof}

\begin{lemma}
\label{lem:PIRPowerParentConvergence2}
Let $c_2$ be a configuration of {\em A2} 
where 
$\#PIR_{c_2}  > 0$.
Let $e$ be a execution from $c_2$.
Let $c_4$ be the configuration reached after $4$ rounds along $e$.
If 
$\#PIR_{c_4} = \#PIR_{c_2}$ then
$\#PIR\_PowerParent_{c_4} < \#PIR\_PowerParent_{c_2}$.
\end{lemma}
\begin{proof}
The proof is similar to the  proof of Lemma \ref{lem:PICPowerParentConvergence2}.

\medskip
In the proof, we assume that $\#PIR_{c_4} = \#PIR_{c_2}$.

\medskip
Let $cs$ be a computation step from $c$ to  $c'$ belonging 
to the first $4$ rounds of $e$.
According to Lemma~\ref{lem:PIRPowerParentClosure}, 
if in $c'$,
we have PIR($u$) then PIR($u$) 
is verified in $c$ and in $c_2$.
According to the the hypothesis, 
any process verifying PIR in $c_2$,
verifies this predicate in any configuration 
reached during the first $4$ rounds of $e$. %:  $\#PIR_{c_4} = \#PIR_{c_2}$. 
According to Lemma~\ref{lem:influentialMove} (page~\pageref{lem:influentialMove}),
the only rule that may execute a process $u^1$ verifying PIR in $c_2$ during the 
the first $4$ rounds of $e$ is R5.
Assume that $u^1$ executes R5 during $cs$.  
In $c$, PowerParent($u^1$) is verified but not in $c'$.
According to Lemma~\ref{lem:PIRPowerParentClosure},
we have $\#PIR\_PowerParent_{c_4} < \#PIR\_PowerParent_{c_2}$.

\medskip
In the sequel, we study executions $e$ 
where processes verifying PIR in $c_2$ are silent during the  first $4$ rounds.

Assume in $c_2$, there is a process $u^2$ verifying $PIR(u^2) ~\wedge~ S.u^2 =Power$ : 
we have $u^2 \neq r$.
According to Theorem~ \ref{theo:power}, $u^2$ executes a rule during 
the first $4$ rounds of $e$.
There is a contraction. We conclude that in $c_2$, we have 
$(PIR(u) \implies PIR\_PowerParent(u))$.
As, processes verifying PIR in $c_2$ are silent during the $4$ first rounds of $e$;
in any reached configuration we have $(PIR(u) \implies PIR\_PowerParent(u))$.
So during the first $4$ rounds, no process verifies PowerConflict and 
Conflict($r$) is not verified as any process $u$ 
verifies $S.u \neq Power ~\vee~ C.u = C.r$. 
Hence, the rule RC3 and RC4 are not executed during the first $4$ rounds of $e$. 

Let $cs$ be a computation step from $c$ to $c'$ 
done during the first $4$ rounds of $e$.
If $S.w = StrongE$ in $c'$ then $S.w = StrongE$ in $c$ and in $c_2$.
Let $c_3$ be the configuration reached at the end of 
the 2th round of $e$ from $c$. 
According to the Lemma~\ref{lem:StrongE} (page~\pageref{lem:StrongE}),
in $c_3$ and during the  first round
from $c_3$ along $e$, every process verifies $S.w \neq StrongE$.

A process $u$ verifying PowerParent($u$) as the $Idle$ status
so QuietSubTree($u$) is verified as no process is faulty
in {\em A1}.
According to the definition of the predicate PowerParent,
in $c_3$ there is a process $u^2$ verifying PIR\_PowerParent($u^2$) and  
$ph.P.u^2 \neq ph.u^2$. Let $v$ be $P.u^2$ in $c_3$.
Till PowerParent($u^2$) and $ph.P.u^2 \neq ph.u^2$ are verified 
the rules RC1, RC2, RC6, R1, R2, R3, R4, R5, 
R6 and R7 are disabled on $v$.
If $v$ executes the rule RC3, RC4, RC5 then in the reached configuration, 
we have $S.v \in Erronneous$ so PowerParent($u^2$) is not verified.
Hence, according to Lemma~\ref{lem:PIRPowerParentClosure}, 
we have $\#PIR\_PowerParent_{c_4} < \#PIR\_PowerParent_{c_2}$.

So, $v$ does not execute any rule  during 
first $2$ rounds of $e$ from $c_3$.
Therefore, NewPhase($u^2$) is verified in $c_3$ and
during the 3th round of $e$ from $c_2$.
Till NewPhase($u^2$) is verified, $u^2$ may execute
a rule (RC4, RC5, R4, or R5) because every process verifies $S.w \neq StrongE$.
The process $u^2$ cannot be silent during the third round of $e$.
We conclude that at least a process verifying PIR in $c_2$, 
does a move during the first $4$ rounds.
So, we have $\#PIR\_PowerParent_{c_4} < \#PIR\_PowerParent_{c_2}$.
\end{proof}

\begin{lemma}
\label{lem:A3-attractor}
{\em A3} is a {\em A2}-attractor reached in at most $8n-8$ rounds 
from a configuration of {\em A2}.
(i.e. {\em A3} is a {\em A0}-attractor reached in at most $16n-15$ rounds 
from any configuration).
\end{lemma}
\begin{proof}
The predicates $\neg$PIR\_PowerParent and $\neg$PIR are closed 
(Lemma~\ref{lem:PIRPowerParentClosure}).

\medskip
\noindent
The sequel of the proof is similar to the convergence part 
of proof of Lemma \ref{lem:A2-attractor}.

Let $c2$ be a configuration of {\em A2}.
$ValPIR_c = \#PIR\_PowerParent_c + \#PIR_c$.
We have $ValPIR_c \leq 2n-2$ as PIR($r$) is not verified.
Let $cs$ be a computation step from $c$ reaching $c'$ done 
during the first $8n-8$ rounds of $e$.
As the predicates $\neg$PIR and  $\neg$PIR\_PowerParent are closed  
along  $e$ (Lemma~\ref{lem:PIRPowerParentClosure}), we have
$ValPIR_c \geq ValPIR_{c'}$.

\noindent
Let $c4$ be the configuration reached the end of 
first $4$ rounds of $e$ from $c2$.
According to Lemma~\ref{lem:PIRPowerParentConvergence2},
we have $ValPIR_{c4} < ValPIR_{c2}$ if $ValPIR_{c2} > 0$.\\
Let $cf$ be the configuration reached the end of first $8n-8$  rounds of $e$ from $c2$.
So, we have $ValPIR_{cf} = 0$.
We conclude that no process verifying PIR($u$).
\end{proof}

\section{Properties of {\em A4}}

\begin{lemma}
\label{lem:A4-attractor}
Let {\em A4} be the set of configurations. 
{\em A4} $\equiv$ {\em A3} $\cap$ \{ $\forall$ $u$,  
$(S.u \neq StrongE$) \} 
is an {\em A3}-attractor reached from $A3$ in $2$ rounds
(i.e. {\em A4} is a {\em A0}-attractor reached in at most $16n-13$ rounds 
from any configuration).
\end{lemma}
\begin{proof}
Let $c3$ be a configuration of {\em A3}. 
Let $u$ be a process verifying 
in $c3$ StrongConflict($u$). 
So two  processes of N[$u$], $v1$ and $v2$, 
has the $Power$ status and do not have the same color. 
One of them, denoted $v$,
verifies the predicate  (influential($v$) 
$\wedge$ unRegular($v$))  in $c3$.  
It is impossible according to the definition of {\em A3}.
Therefore the rule RC4 is not enabled at any process in $c3$. \\
\noindent
Assume that in $c3$, Conflict($r$) is verified.
{\em first case.} in $c3$, 
$r$ has a neighbor $v3a$ having the Power status 
but not the $r$'s color.
{\em second case.} in $c1$, 
$r$ has a neighbor $v3b$ having the Power status 
and $Child.r = \emptyset$ (i.e. no process except $r$ verifies 
the predicate inLegalTree). 
So $v3a$ or $v3b$, named $v$ in the sequel, 
verifies the predicate  
(influential($v$) $\wedge$ unRegular($v$))  in $c3$.
 It is impossible according to the definition
of {\em A3}. Therefore the rule RC3 is not enabled in $c3$. \\
\noindent
We conclude that the property
($S.u \neq StrongE$) is closed in {\em A3}.

\medskip
\noindent
According to Lemma \ref{lem:StrongE} (page~\pageref{lem:StrongE}), 
after $2$ rounds of any execution from $c3$,
a configuration of {\em A4} is reached.
\end{proof}

\begin{lemma}
\label{lem:A4-Conflit}
Let $c4$ be a configuration of {\em A4}.
In $c4$, the rule RC5 is disabled at any process verifying the predicate inLegalTree.
\end{lemma}
\begin{proof}
Assume that  the rule RC5
is enabled at $u$ a process $u$ verifying inLegalTree($u$). 
In $c4 \in A4$,  $u$ cannot verify  one of the following 
predicates : Faulty($u$), PowerFaulty($u$), IllegalLiveRoot($u$), IllegalChild($u$). 
So Conflict($u$)
is verified in $c4$.
So $u$ has a neighbor $v$ verifying $S.v = Power$ 
and $C.v = \overline{r\_color}$ ($r\_color$ being the color of $r$).
We have Influential($v$) 
and $\neg$inLegalTree($v$), in $c4$.
It is impossible according to the definition
of {\em A4}.
\end{proof}

\begin{lemma}
\label{lem:A4-pseudoRegular}
Any execution from a configuration of {\em A4} 
is a pseudo-regular execution (Definition~\ref{def:pseudo-regular}, 
page~\pageref{def:pseudo-regular}).
In {\em A4}, the predicate  $\neg$unRegular($u)$ is closed
\end{lemma}
\begin{proof}
Let $c4$ be a configuration of {\em A4}.
In $c4$, no process is influential and unRegular and no process verifies the predicate
$PowerFaulty$ according to the definition of {\em A4}.
Let $cs$ be a computation step from $c4$ reaching the configuration $c5$.
In $c4$, no process verifies on of the both predicates PowerFaulty and StrongEReady.
So during $cs$, $r$ does not executes the rule  RC1 or RC2.
According to Lemma~\ref{lem:A4-attractor}, $c5 \in A4$.
So during $cs$, no process executes the rule  RC3 or RC4.\\
\noindent
According to Lemma~\ref{lem:A4-Conflit}, the rule RC5 is disabled 
on processes verifying the predicate inLegalTree.
We conclude that during $cs$, processes verifying the predicate
inLegalTree execute only the rules R1 to R7. 
Hence, any execution from $c4$ is a safe execution
(Definition~\ref{def:safe-execution}, page~\pageref{def:pseudo-regular}).

\noindent
During $cs$,  a process $u$ becomes unRegular 
if it executes the rule R3 to choose as a parent an unRegular 
process $v$ having the status $Power$.
We have Influential($v$) 
and $\neg$inLegalTree($v$), in $c4$.
It is impossible according to the definition
of {\em A4}.
We conclude that any execution from $c4$ is a pseudo-regular execution.
\end{proof}

\begin{lemma}
\label{lem:A4-unRegular}
%In {\em A4}, the predicate ($\neg$unRegular($u$) $\wedge$ ($S.u \neq WeakE$))
%is closed for every process $u$.
On any configuration of {\em A4}, the predicate $Ok(r)$ 
is verified.
\end{lemma}
\begin{proof}
%Let $c4$ be a configuration of {\em A4}.
%Let $cs$ be a computation step from $c4$ reaching the configuration $c5$.
%
%\noindent
%According to Lemma~\ref{lem:A4-attractor}, $c5 \in$ {\em A4}.
%So $S.r \notin Erronneous$ in $c5$; hence $\neg$unRegular($r$) is verified  in $c5$.
%
%\noindent
%Let $u$ be a process verifying unRegular($u$) $\vee$ ($S.u = WeakE$) in
%the configuration $c5$. We have $u \neq r$.
%Assume that $u$ does not verify this predicate
% in $c4$, we have ($S.u \neq WeakE$) and $\neg$unRegular($u$).
% 
%\noindent
%So in $c5$, we have $\neg$unRegular($u$) (Lemma~\ref{lem:A4-pseudoRegular}).
%Hence $u$ executes RC5 during $cs$ to get the status $WeakE$.
%But the rule RC5 is enabled only on process verifying the predicate unRegular 
%in $c4$ (Lemma~\ref{lem:A4-Conflit}). 
%There is a contradiction : RC5 is not enabled on $u$ in $c4$.
%
%\medskip
%\noindent
The predicate Faulty($r$), IllegalLiveRoot($r$) and IllegalChild($r$) 
are never verified.
The predicate PowerFaulty($r$) is not verified in $c4 \in A4$.
In $c4$, if the predicate StrongConflict($r$) is verified 
then Conflict($r$) is verified.
Assume that Conflict($r$) is verified in $c4$.
First case, in $c4$, 
$r$ has a neighbor $v1$ having the Power status 
but not the $r$'s color.
Second case, in $c4$, 
$r$ has a neighbor $v2$ having the Power status 
and $Child.r = \emptyset$ (i.e. no process except $r$ verifies 
the predicate inLegalTree). 
So $v1$ or $v2$, named $v$ in the sequel, verifies the predicate 
 (influential($v$) $\wedge$ unRegular($v$))  in $c4$. It is impossible.
We conclude that in $c4$, $Ok(r)$ is verified.
\end{proof}

\begin{lemma}
\label{lem:A4-inLegalTree}
In {\em A4}, the predicate 
$\neg unRegular(u) ~\wedge~(inLegalTree(u) \vee (S.u = Idle))$ 
is closed.
\end{lemma}
\begin{proof}
Let $c4$ be a configuration of {\em A4}.
Let Pre($u$) be the predicate 
unRegular($u$) $\vee$ $(\neg inLegalTree(u) \wedge (S.u \neq Idle))$.
Let $cs$ be a computation step from $c4$ reaching the configuration $c5$.
Let $u$ be a process verifying Pre($u$) in
the configuration $c5$. 
We have $u \neq r$, as inLegalTree($r$) is verified in any configuration of {\em A4}.
Assume that $u$ does not verify Pre($u$) in $c4$, 
we have  $\neg$unRegular($u$).
So in $c5$, we have $\neg$unRegular($u$) 
(Lemma~\ref{lem:A4-pseudoRegular}).

First case, $inLegalTree(u)$ is verified in $c4$ but not in $c5$.
Hence, $u$ executes R7 during $cs$:  in $c5$, we have $S.u=Idle$ :
Pre($u$) is not verified in $c5$.
%TODO preciser dans quel cas inLegalTree n'est plus verifié
Second case, $\neg inLegalTree(u)$ is verified in $c4$.
  Detached($u$) is verified in $c4$, as $\neg$unRegular($u$) is verified.
  By hypothesis $S.u = Idle$ in $c4$.
  The only  R3 may be enabled on $u$ in $c4$.
  Hence, in $c5$, we have $S.u = Idle$.
  Pre($u$) is not verified in $c5$.

We conclude that if in  $c5$, Pre($u$) is verified then it is verified in $c4$.
\end{proof}

\subsection{The termination of the current tree construction in {\em A4}} 

\begin{definition}
The predicate correct($u$) is defined as 
\\ \hspace*{0.2cm}
$[  \neg unRegular(u) ~\wedge$% $(S.u \neq WeakE)~ \wedge $
$(inLegalTree(u) \vee (S.u = Idle) ) ]$
%\\ \hspace*{5cm}
$\wedge$ $((u = r) \vee (dist(u) < dist(TS.u)) $

Let {\em A5($l$)} with $ 0\leq l \leq \diam+1$ be the set of configurations 
of {\em A4} verifying $\forall$ $u \in V$ we have \\
$((dist(u) > l) \vee (C.u = C.r))$ and
 $((dist(u) > l-1) \vee correct(u))$.
\end{definition}

Any configuration where EndLastPhase($r$) is verified belong to $A5(l)$ 
with $l \in [0,\diam+1]$

\noindent
In the sequel, we establish that in at most $n(2n+3)$  
rounds from a configuration of {\em A4}, {\em A5(l)}  
is reached along any execution with $l \geq 0$. 

A tree construction is composed of phases, a root action 
is the beginning of a new phase (R2) or  the beginning of a new tree 
construction (R1). 
A phase has 3 exclusive stages: forwarding, expansion, 
and backwarding. 
The stages are characterized by  the state of processes in the 
legal tree (i.e. processes verifying the predicate inLegalTree). 
During the forwarding, some  influential processes do not have  
the $Power$ status. 
During the expansion, any influential processes have the 
$Power$ status. 
In the backwarding, there is no influential process.

\begin{lemma}
\label{lem:forwarding}
From a configuration $c4$ of {\em A4}, the forwarding stage takes 
at most $n-1$ rounds.
\end{lemma}
\begin{proof}
Let $cs$ be a computation step from $c4$ to $c4'$ 
where $r$ does not execute any rule. 
If in $c4'$, influential($u$) is verified 
then influential($u$) is verified in $c4$ 
(see Lemma \ref{lem:becominginfluential3}, 
page~\pageref{lem:becominginfluential3}).
If in $c4$, $S.u = Power$ is verified 
then the predicate PowerParent($u$) is not verified in $c4'$.
We conclude that in $c4'$, PowerParent($u$) is verified 
then PowerParent($u$) is verified in $c4$.

\noindent
In $c4$, there is at most $n-1$ processes verifying  PowerParent.
Assume that in $c4$ there is $0 < l \leq n-1$ processes 
verifying  PowerParent.

\noindent
A process $u$ verifying PowerParent($u$) as the $Idle$ status
so QuietSubTree($u$) is verified as no process is faulty
in {\em A4}.
According to the definition of the predicate PowerParent,
in $c4$ there is a process $u$ verifying  and 
PowerParent($u$) and $S.P.u = Working$.
So NewPhase($u$) is verified till  PowerParent($u$) is verified.
As in {\em A4}, no processes has the status $StrongE$, and RC4 and RC5 are disabled
on influential processes.
We conclude that R4 or R5 is enabled at $u$ till  
PowerParent($u$) is verified; so after the first round from $c4$ 
there is at most $l-1$ processes verifying  PowerParent
because $u$ does not verify  PowerParent.
Therefore, after $l$ rounds, all influential processes  have the $Power$ status.
\end{proof}
\begin{lemma}
\label{lem:backwarding}
The backwarding stage takes 
at most $n-1$ rounds.
\end{lemma}
\begin{proof}
Let $c$ be a configuration of {\em A4} where no process
is influential. Until a $r$ move, no process is influential  
(Lemma \ref{lem:becominginfluential3}).
$h_{c}$ is the maximal height of a $Working$ process in 
the legal tree in $c$ (we have $h_c < n$).
If $h_c = 0$ then EndPhase($r$) is verified; so, 
it exists $ 0 \leq l \leq \diam$
such that $c \in$ {\em A5(l)}.

\noindent
Assume that $h_c >0$. Let $u$ be a $Working$ process in the legal tree. 
As no process verifies the Predicate Faulty and no process is influent, 
all children of $u$ has the $Working$ or $Idle$ status, 
and they have the same phase value as $u$. 
All $Working$ process in the legal tree at height $h_c$, 
verify the predicate EndPhase because their children have the Idle status. 
They will verify the predicate EndPhase until their next action (R6 or R7) 
so they are enabled. We conclude that 
at $c'$  the configuration reached after a single round 
from $c$. We have $h_{c} >h_{c'}$. 
Therefore, after at most $h_c \leq n-1$ rounds, a configuration 
where EndIntermediatePhase($r$) on EndLastPhase($r$) is verified is reached. 
\end{proof}

\begin{lemma}
\label{lem:r_move}
Whatever is the current configuration of {\em A4} and the execution, 
the root does an action during the first $2n+3$ rounds.
\end{lemma}
\begin{proof}
Lemma \ref{lem:forwarding} establishes that the forwarding stage takes 
at most $n-1$ rounds.
Theorem~\ref{theo:power} (page~\pageref{theo:power}) 
establishes that the expansion 
stage has a duration of at most $4$ rounds.
Lemma \ref{lem:backwarding} establishes that the backwarding stage takes 
at most $n-1$ rounds.
From any configuration after at most $2n+2$ rounds, a configuration $c$
where  EndIntermediatePhase($r$) or EndLastPhase($r$) is verified is reached.
As in a configuration of  {\em A4}, no processes has the status $StrongE$, 
and the predicate $Ok(r)$ is verified (Lemma \ref{lem:A4-unRegular}).
EndIntermediatePhase($r$) or EndLastPhase($r$)  stays verified 
until $r$ executes the rule R1 or R2. 
So $r$ does an action before the end of  the $2n+3$ rounds from any configuration of 
{\em A4}.
\end{proof}

\begin{definition}
Let $\#r\_color(c)$ be the number of processes having the $r\_color$
 in the configuration $c$.
\end{definition}
\begin{lemma}
Let $cs$ be a computation step from $c1 \in A4$ to reach $c2$ 
where $r$ does not executes R1. We have
$\#r\_color(c1) \leq \#r\_color(c2)$.
\end{lemma}
\begin{proof}
$r$ does not change its color (by hypothesis).
Let $u$ be a process that changes its color during $cs$: 
it executes the rule R3 to choose a process, $v$, having the Power status as parent.
As $c1 \in A4$, inLegalTree($v$) is verified in $c1$; 
so $C.v= C.r = r\_color$.
Hence, $C.u = r\_color$ in $c2$.
We conclude that $\#r\_color(c1) \leq \#r\_color(c2)$.
\end{proof}

\begin{lemma}
Let $e$ be a execution starting from a configuration of {\em A4}.
Let $cs = (c1,c2)$ be a computation step of $e$
where $r$  executes R2. 
Let $cs'= (c3,c4)$ the following  computation step of $e$
where $r$  executes a rule.
$\#r\_color(c1) < \#r\_color(c3)$ or $c3 \in A5(l)$ with  $0 \leq l \leq \diam$.
\end{lemma}
\begin{proof}
If along $e$, from $c1$ to $c3$, 
some processes execute R3 then $\#r\_color(c1) < \#r\_color(c3)$.
So assume that no process execute R3 from $c1$ to $c3$.
All processes verifying inLegalTree in $c1$ execute the rule R7.
We conclude that EndLastPhase($r$) is verified in $c3$.
\end{proof}

\begin{theorem}
\label{theo:A4toA5(l)}
Let $c$ be configuration of {\em A4}.
along any execution from $c$
a configuration of $A5(l)$ with  $0 \leq l \leq \diam+1$ is reached
in $n(2n+3)-1$ rounds.
\end{theorem}
\begin{proof}
Let $e$ be an execution starting from a configuration of {\em A4}.
There are at most $n-1$ consecutive computation steps 
of $r$ in which R2 is executed.
At most $2n+2$ rounds are between two actions of $r$ 
(Lemma~\ref{lem:r_move}).
\end{proof}

\subsection{Tree constructions from {\em A5(l)}}

In this subsection, we establish that in at most $n(n+3)$ 
rounds from a configuration of {\em A5(l)} with $l \leq \diam$, 
{\em A5(l+1)} is reached along any execution.

\noindent
For $0 \leq k \leq l$, $A4(k,l)$ is the set
containing all configurations reached during the $k$th 
phase of a tree construction starting 
from a configuration of $A5(l)$.

\begin{definition}
Let {\em A4(k,l)} with $1 \leq k \leq l \leq \diam+1$ 
be the set of 
configurations of {\em A4} where every process 
$u$ verifies the following predicates:

\begin{blist}
\item if $dist(u) < l$ then   correct($u$) 
\item if $dist(u) \leq k-1$ then  $C.u = C.r$  
\item if $k < dist(u) \leq l$  then 
 $C.u \neq C.r$ 
\item if  $dist(u) = k-1$ ~then
\\ \hspace*{1.5cm}
 $(\forall v \in N(u) ~|~ C.v = C.r) $ $\vee$
(influential($u$) $\wedge$  $(Child(u) = \emptyset ~\vee~ S.u = Power))$
\item if $dist(u) = k$ then  
\\ \hspace*{3.5cm}
 $C.u \neq C.r$ $\vee$
 \\ \hspace*{4.5cm}
$[inLegalTree(u) ~\wedge~ \neg influential($u$) ~\wedge~$
 \\ \hspace*{4.5cm}
$S.u=Idle~\wedge~Child(u) = \emptyset~ \wedge~ correct(u) ]$ 
\item if  $k \leq dist(u)$ then   $S.u \in \{Idle, Working, WeakE\}$
\end{blist}
\end{definition}

\begin{observation}
From a configuration of {\em A5(l)}, any computation step 
where $r$ executes the rule R1
reaches a configuration of $A4(1,l)$ verifying
\begin{blist}
\item if $dist(u) < l$ then   correct($u$) 
\item if $1 \leq dist(u) \leq l$  then  $C.u \neq C.r$ 
\item $S.r = Power$
\item if  $u \neq r$ then   $S.u \in \{Idle, Working, WeakE\}$
\end{blist}
%Let $c$ be a configuration where $EndLastPhase(r)$ is  verified,
%$c \in  \bigcup_{1 \leq l \leq \diam}${\em A4(l+1,l)} $\cup$ {\em A4(\diam+1,\diam+1)}.
\end{observation}

\begin{lemma}
\label{lem:correctClosure}
Let $c$ be a configuration of  {\em A4(k,l)} with $1 \leq k \leq l \leq \diam+1$.
Let $cs$ be a computation step from $c$ reaching $c'$.
Let $u$ verifying $dist(u) < l+1$.
If  correct($u$) is verified in $c$ then correct($u$) is verified in $c'$.
\end{lemma}
\begin{proof}
The following predicate is closed according to Lemmas%~\ref{lem:A4-unRegular}, and
~\ref{lem:A4-inLegalTree}:\\
$[  \neg unRegular(u) ~\wedge~% (S.u \neq WeakE)~
 \wedge (inLegalTree(u) \vee (S.u = Idle) ) ]$.

If $dist(u) < dist(TS.u)$ in $c$ then this predicate is verified in $c'$ when $u$
does not execute R3 during $cs$.\\
If $dist(u) >k$, R3 is disabled on $u$ in $c$  because none of its neighbor is influential.
If  $dist(u) < k$,  R3 is disabled on $u$ in $c$ because any influential 
process has the $u$ color.\\
Assume that $dist(u) =k$ and $u$ executes R3 during $cs$. 
In $c'$, we have $dist(TS.u) = k-1$ because in the neighborhood of $u$, 
 only processes at distance $k-1$ of $r$ may be influential in $c$.
\end{proof}
\begin{observation}
Let $c$ be a configuration of  {\em A4} where $u$ is influential.
Let $cs$ be a computation step from $c$ reaching $c'$.
During $cs$, $u$ may execute one of the following rule R1, R2, R5, R6 or R7.
Process $u$ is influential in $c'$, if $u$ does not execute 
the rule R2,  R6 or R7 during $cs$.
\end{observation}

\begin{lemma}
\label{lem:in A4(k,l)}
Let $c$ be a configuration of  {\em A4(k,l)} with $1 \leq k \leq l \leq \diam+1$.
Let $cs$ be a computation step from $c$ reaching $c'$.
If $r$ does not execute any rule then $c'$ is a configuration of  {\em A4(k,l)}.
\end{lemma}
\begin{proof}
$r$ color is the same in $c$ and in $c'$.
In $c'$, we have
\begin{blist}
\item if $dist(u) < l$ then   correct($u$); 
because the predicate   (correct($u$) $\wedge~ dist(u) < l)$ is closed
(Lemma~\ref{lem:correctClosure})
and it is verified in $c$.

\item if $dist(u) \leq k-1$ then  $C.u = C.r$, because R3 is disabled 
on $u$ in $c$, so $u$ keeps its color.
\item if $k < dist(u) \leq l$  then  $C.u \neq C.r$, 
because R3 is disabled on $u$ in $c$, so $u$ keeps its color.

\item assume that $dist(u) = k-1$. if in $c$,  $(\forall v \in N(u) ~|~ C.v = C.r)$
 then this property is verified in $c'$ as R3 is disabled on any $u$' neighbor.\\
Assume that in $c$, $u$ has a neighbor $v$ verifying $(C.v \neq C.r) $. 
As $c$ belongs to {\em A4(k,l)},  influential($u$) is verified.
During $cs$, $u$ may only execute the rule R5, so influential($u$) is verified
in $c'$.
If in $c$, $S.u = Power$ then $u$ is disabled; so $S.u = Power$ in $c'$.
Otherwise $Child(u) = \emptyset~\wedge~S.u \neq Power$ in $c$, (by hypothesis); 
in that case, $Child(u) = \emptyset$ in $c'$.
 
\item Assume that $dist(u) = k$. 
If $(C.u \neq C.r)$ in the configuration $c$ 
then this property is verified in the configuration
$c'$ or $u$ executes R3 during $cs$.
 In the latest case, according to the definition of R3 action, 
 in $c'$, we have $S.u=Idle$, Child($u$) = $\emptyset$, 
 and $\neg$ influential($u$).
According to the definition of {\em A4(k,l)}, 
 in $c'$, we have inLegalTree($u$), $dist(TS.u) =k-1$.
 So, the predicate Pre($u$)  
$\equiv $ 
(inLegalTree($u$) $\wedge$ $\neg$ influential($u$) $\wedge$
$S.u=Idle$ $\wedge$ Child($u$) = $\emptyset$ $\wedge$ correct($u$))
is verified in $c'$.\\
Otherwise, the predicate Pre($u$) is verified  in $c$ by hypothesis.
%Pre($u$) $\equiv $\\
%$(inLegalTree(u) ~\wedge~ \neg influential($u$) ~\wedge~S.u=Idle~\wedge~Child(u) = \emptyset~ \wedge~ correct(u))$.
The predicate  correct($u$) is verified in $c'$ (Lemma~\ref{lem:correctClosure}).
As any execution from $c$ is pseudo-regular (Lemma~\ref{lem:A4-pseudoRegular})
and Pre($u$) is verified in $c$ , Process $u$ is disabled in $c$. So $S.u = Idle$, 
 is verified in $c'$.
In {\em A4}, the predicate inLegalTree($u$) is verified until $u$ executes R7.
$\neg$ influential($u$) is verified in $c'$ as $\neg$ influential($u$) is verified in $c$ 
(Lemma~\ref{lem:becominginfluential3}, page~\pageref{lem:becominginfluential3}).
$S.u = Idle ~\wedge~Child(u) = \emptyset$ is verified in $c$, hence 
$Child(u) = \emptyset$ is verified in $c'$.
We conclude that Pre($u$) is verified  in $c'$.

\item if  $k \leq dist(u)$ then   $S.u \in \{Idle, Working, WeakE\}$. 
The processes at distance $k$, do not have child if they are in the legal tree.
We conclude that
the predicate inLegalTree($u$) is not verified; so $u$ is be influential as $c \in A4$.
\end{blist}
\end{proof}

\begin{observation}
\label{obs:rmove}
Let $c$ be a configuration of {\em A4} where $r$ is enabled.
There is not influential process, so  R3, R4, and  R5 are disabled on any process in $c$.
A process $u$ verifying inLegalTree($u$) is disabled in $c$ and
$S.u = Idle$, $ph.u = ph.r$.
\end{observation}

\begin{lemma}
\label{lem:A4(k,l)}
Let $cs$ be a computation step from $c(k) \in A4(k,l)$ with $1 \leq k < l \leq \diam+1$
% to $c(k+1) \in Al(k+1,l)$ with $1 \leq k < l$~: during $cs$, 
where $r$ executes a rule.
The reached configuration $c$ belongs to $(k+1,l)$ and $r$ executes the rule R2.
\end{lemma}
\begin{proof}
 $c(k) \in A4(k,l)$ verifies the predicate :
\begin{blist}
\item   EndIntermediatePhase($r$)
\item if $dist(u) < l$ then   correct($u$) 
\item if $dist(u) \leq k$ then  $C.u = C.r$  
\item if $k < dist(u) \leq l$  then  $C.u \neq C.r$ 
\item if $dist(u) = k$ then   
\\ \hspace{2cm}
$(inLegalTree(u) ~\wedge~ \neg influential($u$) ~\wedge~S.u = Idle ~\wedge~ Child(u) = \emptyset~ \wedge~ correct(u))$ 
\item if  $k+1 \leq dist(u)$ then   $S.u \in \{Idle, Working, WeakE\}$
\end{blist}
The predicate  correct($u$) is verified in $c$ (Lemma~\ref{lem:correctClosure}).
During $cs$, a process $u \neq r$ may execute only the rule R6, R7, RC5, or RC6
(Observation~\ref{obs:rmove}).
During $cs$, $r$ executes $R2$.
A process $u$ verifying inLegalTree($u$) in $c(k)$ is influential in $c$ (Lemma~\ref{lem:becominginfluential3}, page~\pageref{lem:becominginfluential3}).
%TODO to proof ``A process $u$ verifying inLegalTree($u$) in $c(k)$ is influential in $c$''
\end{proof}

We define $A4(l+1,l)$ such that this set
contains all configurations reached during the last 
phases of a tree construction starting 
from a configuration of $A5(l)$; 
more precisely, the latest phases are phases $k$ with ($ k > l$).

\begin{definition}
Let {\em A4(l+1,l)} with $l \leq \diam$ be the set of 
configurations of {\em A4} where every process 
$u$ verifies the following predicates:
\begin{blist} 
\item if $dist(u) \leq l$ then  correct($u$) and $C.u = C.r$  
\item if  $dist(u) = l$ and $(\exists v \in N(u) ~|~ C.v \neq C.r) $
then \\ \hspace*{5cm}
Influential($u$) and 
 $(Child(u) = \emptyset ~or~ S.u = Power)$
\end{blist}
\end{definition}

The proof of the following lemma is similar to the proof of lemma~\ref{lem:in A4(k,l)}.
\begin{lemma}
\label{lem:moveInA4(l+1,l)}
Let $c$ be a configuration of  {\em A4(l+1,l)} with $ l \leq \diam$.
Let $cs$ be a computation step from $c$ reaching $c'$.
If $r$ does not execute R1 then $c'$ is a configuration of  {\em A4(l+1,l)}.
\end{lemma}
\begin{proof}
$r$ color is the same in $c$ and in $c'$.
In $c'$, we have
\begin{blist}
\item by hypothesis, correct($u$) is verified in $c$.
According to Lemma%~\ref{lem:A4-unRegular}, and
~\ref{lem:A4-inLegalTree},
if $dist(u) < dist(TS.u)$ in $c$ then correct($u$) is verified in $c'$ in the case where $u$
does not execute R3 during $cs$.
If  $dist(u) \leq l$,  R3 is disabled on $u$ in $c$ because any influential 
process has the $u$ color.
\item if $dist(u) \leq l$ then  $C.u = C.r$, because R3 is disabled 
on $u$ in $c$, so $u$ keeps its color.
\item   assume that $dist(u) = l$. 
If in $c$,  $(\forall v \in N(u) ~|~ C.v = C.r)$ then 
this property is verified in $c'$ as R3 are disabled on any $u$' neighbor..\\
Assume that in $c$, $u$ has a neighbor $v$ verifying $(C.v \neq C.r) $. 
As $c$ belongs to {\em A4(l+1,l)},  influential($u$) is verified.
During $cs$, $u$ may only execute the rule R5, so influential($u$) is verified
in $c'$.
If in $c$, $S.u = Power$ then $u$ is disabled; so $S.u = Power$ in $c'$.
Otherwise $Child(u) = \emptyset~\wedge~S.u \neq Power$ in $c$,
 (by hypothesis); in that case
 $Child(u) = \emptyset$ in $c'$.
\end{blist}
\end{proof}

\begin{lemma}
\label{lem:R2moveInA4(l,l)}
Let $cs$ be a computation step from $c(l) \in A4(l, l)$ with $  l \leq \diam$ 
where $r$ executes a rule.
$r$ executes R2 and
the reached configuration $c$ belongs to $A4(l+1,l)$.
\end{lemma}
\begin{proof} 
$c(l) \in A4(l,l)$ verifies the predicate :
\begin{blist}
\item   EndIntermediatePhase($r$)
\item if $dist(u) \leq l$ then   correct($u$) 
\item if $dist(u) \leq l$ then  $C.u = C.r$   
\item if $dist(u) = l$ then   
\\ \hspace{2cm}
$(inLegalTree(u) ~\wedge~ \neg influential($u$) ~\wedge~S.u = Idle ~\wedge~ Child(u) = \emptyset~ \wedge~ correct(u))$ 
\item if  $l+1 \leq dist(u)$ then   $S.u \in \{Idle, Working, WeakE\}$
\end{blist}
%\begin{blist}
During $cs$, a process $u \neq r$ may execute only the rule R6, R7, RC5, or RC6
(Observation~\ref{obs:rmove}).
%TODO a completer
\end{proof}

\begin{lemma}
\label{lem:rmoveInd+1}
Let $cs$ be a computation step from $c(l+1) \in A4(l+1,l) ~\cup ~ A4(\diam+1, _diam+1)$ 
with $  l \leq \diam+1$
where $r$ executes the rule R1.
$c(l+1)$ belongs to $A5(l+1)$  if $c(l+1) \in A4(l+1,l)$ 
otherwise  $c(l+1)$ belongs to $A5(\diam+1)$.
%The reached configuration $c$ belongs to $A4(1,l+1)$ if $c(l+1) \in A4(l+1,l)$
%or it  belongs to $A4(1,\diam+1)$, if $c(l+1) \in A4(\diam+1,\diam+1)$ .
\end{lemma}

\begin{proof} 
If $c(l+1) \in A4(l+1,l)$ it verifies the properties : 
\begin{blist}
\item   EndLastPhase($r$)
\item if $dist(u) \leq l$ then   correct($u$) 
\item if $dist(u) \leq l+1$ then  $C.u = C.r$  
\item if  $l+1 \leq dist(u)$ then   $S.u \in \{Idle, Working, WeakE\}$
\end{blist}
If $c(l+1) \in A4(\diam+1,\diam+1))$ it verifies the properties : 
\begin{blist}
\item   EndLastPhase($r$)
\item if $dist(u) \leq \diam$ then   correct($u$) and  $C.u = C.r$  
\end{blist}
\end{proof}

%\begin{lemma}
%Let $cs$ be a computation step from $c \in A4(\diam,\diam+1)$ 
%where $r$ executes a rule.
%The reached configuration $c'$ belongs to $(\diam+1,\diam+1)$.
%\end{lemma}
%\begin{proof} 
%$c  \in A4(\diam,\diam+1)$ verifies the predicate : 
%\begin{blist}
%\item   EndIntermediatePhase($r$)
%\item $\forall u$ we have   correct($u$) and   $C.u = C.r$  
%\item if $dist(u) = \diam$ then   
%$(inLegalTree(u) ~\wedge~ \neg influential($u$) ~\wedge~Child(u) = \emptyset~ \wedge~ correct(u))$ 
%\end{blist}
%\end{proof}
%\begin{lemma}
%\label{lem:A4(\diam+1,\diam+1)}
%Let $cs$ be a computation step from $c(\diam+1) \in A4(\diam+1,\diam+1)$ 
%where $r$ executes a rule.
%The reached configuration $c$ belongs to $A4(1,\diam+1)$.
%\end{lemma}
%\begin{proof} 
%$c(\diam+1) \in A5(\diam+1)$ verifies the predicate : 
%\begin{blist}
%\item   EndLastPhase($r$)
%\item $\forall u$  we have correct($u$)  and  $C.u = C.r$  
%\end{blist}
%\end{proof}

\subsection{Legitimate configurations}

\noindent
%Let {\em Al(k)= A4(k,\diam+1)} with $1 \leq k \leq \diam+1$.
Let {\em Al} be the set of configurations defined as
$\bigcup_{u \in[0,\diam+1]}A4(k,\diam+1)$

\noindent
$A4(k,\diam+1)$ contains all reached configurations 
during the $k$th phase of a
tree construction starting from a  configuration of $Al$.

According to the  definition of predicate correct, and 
Lemmas~\ref{lem:in A4(k,l)},~\ref{lem:A4(k,l)}, and~\ref{lem:rmoveInd+1}.
\begin{corollary}
{\em Al} is closed. 

Let $c$ be a configuration in  $Al$. In $c$,  for every process $u \neq r$ we have 
$dist(u) > dist(TS.u) $.
\end{corollary}

\begin{theorem}
{\em Al} is reached from a configuration of{\em A4} after at most $(\diam+2)n(2n+3)$ rounds.
\end{theorem}
\begin{proof}
According to Theorem~\ref{theo:A4toA5(l)}, in at most $n(2n+3)-1$ rounds from a configuration of {\em A4}, a configuration of
{\em A5(l)} with $0 \leq l \leq \diam+1$ is reached along any execution.

From a configuration of {\em A5(l)} with $0 \leq l \leq \diam+1$ a configuration of $A4(1,l)$
is reached after a round.

Assume that $l < \diam+1$.  According to Theorem~\ref{theo:A4toA5(l)}, in at most $n(2n+3)-1$ rounds from a configuration of $A4(1,l)$ with $0 \leq l < \diam+1$, a configuration of
{\em A5(l')}, $c(l')$,  is reached. 
According to lemmas~\ref{lem:in A4(k,l)},~\ref{lem:A4(k,l)},~\ref{lem:moveInA4(l+1,l)},
and~\ref{lem:R2moveInA4(l,l)},
$c(l') \in A4(l+1,l)$. We have $l'=l+1$
(Lemma~\ref{lem:rmoveInd+1}).
So, in at most $n(2n+3)-1$ rounds from a configuration of $A4(1,l)$ with $l < \diam$, 
{\em A5(l+1)} is reached along any execution.
Hence, from a configuration of $A4(1,l)$ with $l < \diam+1$, a configuration of
$A4(1,l+1)$ is reached in at most $n(2n+3)$ rounds.
We conclude that from  a configuration of $A4(1,l)$ with $0 \leq l < \diam+1$, a configuration of
$A4(1,\diam+1)$ is reached in at most $(\diam+1)n(2n+3)$ rounds.
Hence, from a configuration of {\em A4}, a configuration of
{\em Al} is reached along any execution in at most $(\diam+2)n(2n+3)$ rounds.
\end{proof}

\begin{corollary}
{\em Al} is reached from any configuration  after at most $2\diam n^2+ 4n^2+ O(\diam n)$ rounds.
\end{corollary}

\noindent
Any execution from a configuration of {\em Al} is a regular 
execution. 

\begin{theorem}
In {\em Al} a complete
BFS tree construction is done in  $\diam^2+3\diam+1$ rounds.
A process executes at most $2\diam+1$ moves during a complete
BFS tree construction.
\end{theorem}
\begin{proof}
We have {\em A5($\diam+1$)} $\subset$ {\em Al}.
From a configuration of {\em A5($\diam+1$)}, 
a tree construction is done in $\diam+1$ phase.
The first phase is done in $2$ rounds : $r$ executes the rule R1, 
then $r$ neighbors (if they exists) execute the rule R3.
At the end of the first phase, all processes at distance $1$ 
of $r$ are $Idle$ leaves of the legal tree.
The $k$th phase takes  $2k$ rounds if $1 < k \leq \diam$:
\begin{list}{}{\setlength{\partopsep}{0.00in} 
\setlength{\topsep}{0.00in} 
\setlength{\itemsep}{0.04in}\setlength{\parsep}{0.00in}
\setlength{\leftmargin}{0.1in}} 
\item $\bullet$
during the first round of the phase, $r$ execute R2, 
\item $\bullet$
during the round $2 \leq i < k$, inLegalTree processes at 
distance $i-1$ of $r$ executes the rule R4, 
\item $\bullet$
during the round $k$, processes at distance $k-1$ of $r$ 
executes R5,
\item $\bullet$
during the round $k+1$, processes at distance $k$ of $r$ 
execute R3,
\item $\bullet$
during the round $k+2 \leq i+k+1 \leq 2k$, inLegalTree 
processes at distance $k-i$ of $r$ executes 
the rule R7 or R6.
\end{list}

\medskip
\noindent
At the end of  the phase $k\leq \diam$, all processes at 
distance $k$ of $r$  are the $Idle$ leaves of the legal tree.
The $\diam+1$th phase takes  $2\diam+1$ rounds :
\begin{list}{}{\setlength{\partopsep}{0.00in} 
\setlength{\topsep}{0.00in} 
\setlength{\itemsep}{0.04in}\setlength{\parsep}{0.00in}
\setlength{\leftmargin}{0.1in}} 
\item $\bullet$
during the first round of the phase, $r$ execute R2, 
\item $\bullet$
during the round $2 \leq i < \diam+1$, inLegalTree processes 
at distance $i-1$ of $r$ executes the rule R4, 
\item $\bullet$
during the round $\diam+1$, processes at distance $\diam$ of $r$ 
executes R5,
\item $\bullet$
during the round $\diam+2 \leq i+\diam+2 \leq 2\diam+1$, inLegalTree
processes at distance $\diam-i$ of $r$ executes the rule R7 
(they quit the legal tree).
\end{list}

\medskip
\noindent
At the end of the phase $\diam+1$ a configuration of {\em A5($\diam+1$)} is 
reached.
A tree construction is done in  $(\diam+2)(\diam+1)-1$ rounds.

\medskip
\noindent
From a configuration of {\em A5($\diam+1$)}, during a single tree construction 
the move of $r$ belongs to language  $R1(R2)^\diam$; and the moves of 
$u$ at distance $l$ of $r$ belongs to the language 
$R3R5(R6R4)^{i}R7$ with $0 \leq i \leq \diam-l$. 
\end{proof}

\bibliography{biblio}
\end{document}